\documentclass[12pt]{article}
\usepackage[breaklinks=true]{hyperref}
\hypersetup{colorlinks=true, linkcolor=black, citecolor=black, urlcolor=blue!80!black} 
\usepackage{amsfonts,amsmath,amsxtra,amsthm}
\usepackage{amssymb}
\usepackage{bm}
\usepackage[utf8]{inputenc}
\usepackage{lscape}
\usepackage[normalem]{ulem}
\usepackage[russian,english]{babel}
\usepackage{xcolor}
\usepackage{pst-circ}
\usepackage{tikz, pgfplots}
\usepackage{tcolorbox}
\usepackage{mathrsfs}
\usepackage{cancel}

\usepackage{tocloft} 
\renewcommand{\epsilon}{\varepsilon}

\numberwithin{equation}{section}

\renewcommand{\phi}{\varphi}

\newtheorem{lemma}{Lemma}
\newtheorem{theorem}{Theorem}
\newtheorem{proposition}{Proposition}
\newtheorem{corollary}{Corollary}
\newtheorem{remark}{Remark}

\renewcommand{\epsilon}{\varepsilon}
\renewcommand{\l}{\lambda}
\newcommand{\bl}{\bm{\lambda}}

\def\bo{\boldsymbol{\omega}}
\newcommand{\bx}{\bm{x}}
\newcommand{\by}{\bm{y}}
\newcommand{\bz}{\bm{z}}
\def\beq{\begin{equation}}
	\def\eeq{\end{equation}}
\def\beqq{\begin{equation*}}
	\def\eeqq{\end{equation*}}
\def\C{\mathbb{C}}

\def\R{\mathbb{R}}

\def\a{\alpha}

\def\l{\lambda}

\def\o{\omega}
\def\SC{\mathcal{S}}

\def\T{T}
\def\TT{T^\dagger}
\def\ve{\varepsilon}
\def\vf{\varphi}

\let\Re\relax
\let\Im\relax
\DeclareMathOperator{\Re}{Re}
\DeclareMathOperator{\Im}{Im}

\newcommand{\rf}[1]{(\ref{#1})}
\def\S{\mathcal{S}}

\begin{document}
	
	\begin{center}
		{\bf \Large Ruijsenaars spectral transform}
		
		\vspace{0.4cm}
			
		{N. Belousov$^{\dagger}$, S. Khoroshkin$^{\circ\ast}$}
		
		\vspace{0.4cm}
		
		{\small \it
		$^\dagger$Beijing Institute of Mathematical Sciences and Applications, \\
		Huairou district, Beijing, 101408, China \vspace{0.2cm}\\
		$^\circ$National Research University Higher School of Economics, \\
		Myasnitskaya 20, Moscow, 101000, Russia \vspace{0.2cm}\\
		$^\ast$Department of Mathematics, Technion,
		Haifa, Israel
		}
	\end{center}
	
	\begin{abstract} 
		Spectral decomposition with respect to the wave functions of Ruijsenaars hyperbolic system defines an integral transform, which generalizes classical Fourier integral. For a certain class of analytical symmetric functions we prove inversion formula and orthogonality relations, valid for complex valued parameters of the system. Besides, we study four regimes of unitarity, when this transform defines isomorphisms of the corresponding $L_2$ spaces.
	\end{abstract}
			
	\vspace{-0.3cm}
			
	\tableofcontents
	
	\newpage
			
	\section{Introduction} 
			{\bf 1.} Hamiltonians of  Ruijsenaars hyperbolic system \cite{R1, HR1} and Macdonald  operators~\cite{M} can be regarded as noncompact and compact incarnations of the same family of commuting difference operators
			\begin{align}\label{H}
				H_s= \sum_{ \substack{J \subset \{1, \dots, n\} \\[1pt] |J| = s}} \, \prod_{\substack{j \in J \\ k\not\in J}} \frac{\sh \frac{\pi}{\omega_2} ( x_j - x_k - \imath g)}{ \sh \frac{\pi}{\omega_2} (x_j - x_k) } \, \prod_{j \in J} e^{- \imath \omega_1 \partial_{x_j}}. 
			\end{align}
			Both Macdonald polynomials and wave functions of Ruijsenaars hyperbolic system admit  explicit representations by iterative contour integrals \cite{AOS, HR1}. Fundamental tools in the theory of Macdonald symmetric functions~--- reproducing kernel \cite[VI.2.7]{M}, two invariant scalar products and Cauchy identity~\cite[VI.4.13]{M}~--- play the role analogous to the kernel and measure functions \cite[(1.16), (2.1)]{HR1}, relations of orthogonality and completeness~\cite[(1.68), (1.69)]{BDKK3} in Ruijsenaars system.  
			
			 Many algebraic results in Macdonald theory use symmetry of Macdonald operators with respect to the bilinear form (``another scalar product'', as it is called in Macdonald's book), which in the notation of \cite[VI.9]{M} is given by the formula
			\beqq
			 \bigl( \psi(\bz),\phi(\bz) \bigr)=\frac{1}{n! \, (2\pi)^n}\int_{|z_i|=1} \prod_{i = 1}^n\frac{dz_i}{z_i} \, \frac{\prod_{i\not=j}(z_i/z_j;q)_\infty}{\prod_{i\not=j}(tz_i/z_j;q)_\infty} \, \psi(\bz^{-1}) \, \phi(\bz),\eeqq
			where $(z; q)_{\infty} = \prod_{k = 0}^\infty(1 -q^k z)$. Here and below we denote the tuples of $n$ variables
			\beqq \bz=(z_1,\ldots,z_n),\qquad \bz^{-1}=(z_1^{-1},\ldots,z_n^{-1}).\eeqq
			Contrary to Macdonald theory, in Ruijsenaars hyperbolic system we deal with functional spaces. Analytical techniques of $L_2$ spaces force one to promote Hamiltonians~\eqref{H} to formally self-adjoint operators with corresponding restrictions on the parameters (e.g. choose real periods $\o_i$ and coupling constant $g$).
			 
			However, most of the wave functions properties can be proved for complex valued parameters $\omega_i, g$ using symmetry of Hamiltonians with respect to the bilinear form
			\beq \label{f0a} \bigl(\psi(\bx),\phi(\bx) \bigr)_{\mu}=\int_{\R^n}d\bx \; \mu(\bx) \, \psi(-\bx) \, \phi(\bx) 
			\eeq
			instead of hermitian scalar product \cite{BDKK2, BDKK3, BDKK4}. Here the measure function $\mu(\bx)$ is given by \rf{f00}. 
			
			The exceptions so far were orthogonality and completeness of the wave functions, which were proved in \cite{BDKK3} only for real $\omega_i, g$. One of the goals of this note is to extend this result to complex values of parameters and then apply it to other unitarity regimes of Ruijsenaars system, which were not considered before. 
			 
			\smallskip
			{\bf 2.} The wave functions $\Psi_{\bl}(\bx)$ of Ruijsenaars hyperbolic system with arbitrary number of particles were constructed in \cite{HR1} and further studied in \cite{HR2, HR3, BDKK2, BDKK3, BDKK4} (see also~\cite[Section~3]{BKK} and references therein for the appearance of these functions in supersymmetric gauge theories). The pairing  \rf{f0a} with the wave function $\Psi_{\bl}(\bx)$ defines the linear map~$\T$
			\begin{equation*}
				[\T \phi] (\bm{\lambda}) = \int_{\mathbb{R}^n} d\bm{x} \; \mu(\bm{x}) \, \Psi_{\bm{\lambda}}(-\bm{x}) \, \phi(\bm{x})
			\end{equation*} 
			from the certain space of functions of spatial variables~$\bx$ to functions of spectral variables~$\bl$, which we call \textit{Ruijsenaars spectral transform}. For one particle case the wave function is a plane wave
			\beqq \Psi_\l(x)=e^{2\pi \imath \l x} \eeqq
			and $T$ is the Fourier transform (or Laplace transform for other space of initial functions). 
			In general case, regarding the map $T$ as a multidimensional analog of Fourier transform, we establish basic properties of Ruijsenaars spectral transform: inversion formula and its extension to the isomorphism of corresponding $L_2$ spaces in four unitarity regimes.

			\smallskip
			{\bf 3.} Following the ideas of Fourier transform, we introduce an analog of  Schwartz space~$\S_{\bo,g}$ that depends on three parameters of the system: periods $\bo=(\o_1,\o_2)$ and coupling constant $g$. It consists of the functions $\phi(\bx)$ analytic in the strips
			$$|\Im x_j|\leq \frac{\min(\Re \omega_1, \Re\omega_2)}{2}, \qquad j = 1, \dots, n, $$ 
			and decaying in these strips, so that
			\beqq  	| \phi(\bm{x}) | \leq C 	 \exp \biggl(- 2\pi \Re \frac{g}{\o_1\o_2} \sum_{1 \leq j < k \leq n} | \Re (x_j - x_k)|-\ve\sum_{j = 1}^n | \Re x_j|\biggr)
			\eeqq 
			for some $\ve>0$. Using the bound on the wave function $\Psi_{\bl}(\bx)$ from \cite{BDKK2} we show that Ruijsenaars spectral transform $\T$ is well defined on this space and is given by absolutely convergent integral.
			
			Similarly to the Fourier transform we then need to estimate decay of the image $[\T\vf](\bl)$ when $|\bl|$ tends to infinity ($\bl \in \mathbb{R}^n$). In Fourier transform theory this can be done using differentiation and integration by parts.  Here we explore instead  difference operators diagonalized by the wave functions $\Psi_{\bl}(\bx)$. Namely, using the whole family of commuting Ruijsenaars--Macdonald operators~\eqref{H} we compose from them the generating function~$H(\l)$, see~\rf{f18}, whose spectrum is separated from zero and tends to infinity when $| \bl |$ increases. Using symmetry of~$H(\lambda)$ with respect to the bilinear form  \rf{f0a}  we derive bound on the image~$[\T\vf](\bl)$ and show that it is in the region of definition of the left inverse operator~$\TT$ given by
			\begin{align}
				[\TT \chi] (\bm{x}) & = \int_{\mathbb{R}^n} d\bm{\lambda} \; \hat{\mu}(\bm{\lambda}) \, \Psi_{\bm{\lambda}}(\bm{x}) \, \chi(\bm{\lambda}).
			\end{align}
		Here $\hat{\mu}(\bm{\lambda})$ is the dual measure function, see \rf{f0b}.
		
		We then use the regularization of the bilinear pairing
		$$ \bigl(\Psi_{\bl}(\by),\Psi_{\bl}(\bx) \bigr)_{\hat{\mu}}$$  
		by means of the relation
		\beq\label{pr1}
			\bigl( \Psi_{\bm{\l}}(\by) ,\, \Psi_{\bm{\lambda}}(\bx) \bigr)_{\hat{\mu}}^{\l, \varepsilon}=\bigl(\Psi_{\bl}(\by),\Psi_{\bl}(\bx)\cdot 	R_{\l, \epsilon}(\bm{\l}) \bigr)_{\hat{\mu}}
		\eeq
		where the function $R_{\l, \epsilon}(\bm{\l}) $ is given by the formula \rf{R}. It was proved in \cite[Proposition 2]{BDKK3} that the regularized pairing \rf{pr1} forms a delta sequence in the Schwartz space
		\begin{align*}
			\lim_{\lambda \to \infty} \, \lim_{\varepsilon \to 0^+} \bigl( \Psi_{\bm{\l}}(\by),\, \Psi_{\bm{\lambda}}(\bx) \bigr)_{\hat{\mu}}^{\l, \varepsilon} = \mu^{-1}(\bm{x}) \, \delta(\bm{x}, \bm{y}).
		\end{align*} 
		Its application gives now both inversion formula (Theorem~\ref{theoremf1})
		 \beqq [\TT\T  \vf] (\bx)=\vf(\bx), \qquad \vf\in\SC_{\bo,g}, \qquad \bx\in\R^n\eeqq
		and equivariance property of the spectral transform (Theorem~\ref{theoremf2})
		  \beqq \bigl([\T\vf_1](\bl),\,[\T\vf_2](\bl)\bigr)_{\hat{\mu}}= \bigl(\vf_1(-\bx),\vf_2(\bx)\bigr)_\mu \, ,\qquad \vf_1, \vf_2\in\SC_{\bo,g}.\eeqq 
		The latter can be regarded as a generalization of the convolution property of Laplace or Mellin transforms (see e.g.~\cite[(2.4.8)]{AAR}).

		\smallskip	     
		{\bf 4.} Finally, we apply these results to four unitarity regimes of Ruijsenaars hyperbolic system, see Theorem~\ref{theoremf3}. In the most known case, studied by Ruijsenaars \cite[Section~2.2]{R3}, the periods~$\omega_i$ and coupling constant~$g$ are real. In this case the bilinear pairing and hermitian scalar product of the wave functions coincide as well as the corresponding spectral transforms. This implies that the operators $T$ and $T^\dagger$ are isometries and by a standard procedure we deduce that in this case they extend to unitary isomorphisms between spaces $L_2^{\mathrm{sym}}(\mathbb{R}^n, \mu)$ and $L_2^{\mathrm{sym}}(\mathbb{R}^n, \hat{\mu})$. The same arguments are applied to the case of complex conjugated periods~$\bar{\omega}_1 = \omega_2$ and real coupling constant $g$.  
		
		Let us remark that for real $\omega_i, g$ unitarity of spectral transform was proved in \cite{BDKK3} in a different way, following the arguments presented in~\cite{DKM}. The crucial point was the positivity of regularization function $R_{\l,\epsilon}(\bm{\l})$~\eqref{pr1}, which allows to use Fatou's lemma. This proof can be continued to the case of~complex conjugated periods and real coupling, see Section~\ref{sec:fatou}.
			
		The other two regimes of unitarity appear when coupling constant is complex conjugated to its reflection
		\begin{equation}\label{g-restr}
			\bar{g} = \omega_1 + \omega_2 - g.
		\end{equation}
		As it was observed in \cite{DKKSS}, in this case Hamiltonians are symmetric with respect to the scalar product with Sklyanin measure $\Delta(\bm{x})$, which does not depend on the coupling constant, see \eqref{Skl}. In this regime coupling constant symmetry of the wave functions~\rf{refl} proved in \cite{BDKK4} plays the crucial role. With its help we establish a simple connection between spectral transform with respect to bilinear form and with respect to scalar product with Sklyanin measure, see~\rf{f55}. Using this relation we again derive the inversion formula and prove unitary isomorphism of corresponding $L_2$ spaces (Theorem~\ref{theoremf3}). 
		
		We note in passing that restriction of the type~\eqref{g-restr} and Sklyanin measure $\Delta(\bm{x})$ appear in the modular $XXZ$ and $q$-Toda chains~\cite{KLS},~\cite{DKM1},~\cite{SchSh}, as well as in Liouville field theory, where the kernel of $S$-move~\cite[Section~4.4]{T},~\cite[Appendix D]{TV} coincides with the Halln\"as--Ruijsenaars wave functions in the case of two particles. 
			
		To our knowledge unitarity regimes of Ruijsenaars model with complex conjugated periods $\bar{\omega}_1 = \omega_2$ were not considered before. However, such regime has been considered in the context of modular double of quantum group \cite{F1, F2, KLS, BT}. Besides, the strong clue for it is given by the recent work \cite{SS}, where \mbox{$\omega_1/\omega_2 \to -1$} limit of the Ruijsenaars hyperbolic system is considered. Namely, in this limit the periods become pure imaginary and complex conjugated.
		
		Note that in the ``nonrelativistic'' limit ($g = \beta \omega_1$, $\omega_1 \to 0$), which corresponds to hyperbolic Calogero--Sutherland model, the inversion formula as well as the Plancherel theorem are known in a stronger version, see \cite{O}. 
			
	\section{Notations} 
		The Ruijsenaars hyperbolic system depends on three parameters: periods $\bo=(\o_1,\o_2)$ and coupling constant $g$. All these parameters are assumed to be complex valued with restrictions on the periods
		\begin{align}\label{o-cond}
			\Re \omega_1 >0, \qquad  \Re \omega_2 > 0
		\end{align}
		and on the coupling constant
		\begin{align}\label{g-cond}
			0 < \Re g < \Re (\omega_1 + \omega_2), \qquad 0 < \Re \frac{g}{\omega_1 \omega_2} < \Re \frac{\omega_1 + \omega_2}{\omega_1 \omega_2}.
		\end{align}
	All further results are symmetric with respect to interchange of the periods. For definiteness, we assume that 
	\begin{equation}\label{o-ord}
	 	\Re \o_1\leq \Re \o_2.
	 \end{equation} 
	Denote also the reflected coupling constant
	 \begin{align}
	 	g^* = \omega_1 + \omega_2 - g.
	 \end{align} 
	The measure function $\mu(x)$ and the kernel function $K(x)$ 
	\begin{align}
		& \mu(x) \equiv \mu(x; g | \bm{\omega}) = S_2(\imath x | \bm{\omega}) \, S^{-1}_2(\imath x + g| \bm{\omega}), \\[6pt]
		&  K(x) \equiv K(x; g| \bm{\omega}) = S_2^{-1} \biggl(\imath x+\frac{g^*}{2} \bigg| \bo \! \biggr) \, S_2^{-1}\biggl(- \imath x+\frac{g^*}{2} \bigg| \bo \! \biggr)
	\end{align}
	constitute the main building blocks of wave functions. Here $S_2(z | \bo)$ is the double sine function, whose properties are collected in Appendix \ref{AppendixA}. In particular, from reflection formula for the double sine function \eqref{S2-refl} we have the relation
	\begin{align}
		K\biggl(x + \frac{\imath g^*}{2} \biggr) = \mu^{-1}(x).
	\end{align}
	Denote also the measure function for $n$ variables
	\begin{align}\label{f00}
		\mu(\bm{x}) = \frac{1}{n!}\prod_{\substack{j, k = 1 \\ j \not=k} }^n \mu(x_j - x_k),\qquad \bx=(x_1,\ldots,x_n).
	\end{align}
	It admits two natural factorizations. First, it can be factorized into two ``halfs''
	\begin{align}\label{mu-half}
		\mu(\bm{x}) = \mu'(\bm{x}) \, \mu'(-\bm{x}), \qquad \mu'(\bm{x}) = \frac{1}{\sqrt{n!}} \prod_{1\leq j < k \leq n} \mu(x_j - x_k).
	\end{align}
	Second, we can separate its coupling dependent part 
	\begin{align}\label{eta}
		\mu(\bm{x}) = \Delta(\bm{x}) \, \eta(\bm{x}), \qquad \eta(\bm{x}) = \prod_{\substack{j, k =1 \\ j \not= k}}^n S_2^{-1}( \imath x_j- \imath x_k + g | \bm{\omega}).
	\end{align}
	The remaining part has the form
	\begin{align}\label{Skl}
		\Delta(\bm{x}) = \frac{1}{n!}  \prod_{\substack{j, k =1 \\ j \not= k}}^n S_2( \imath x_j- \imath x_k | \bm{\omega}) = \frac{1}{n!} \prod_{1 \leq j < k \leq n} 4 \sh \frac{\pi (x_j - x_k)}{\omega_1} \, \sh \frac{\pi (x_j - x_k)}{\omega_2},
	\end{align}
	where the second expression follows from reflection formula~\eqref{S2-sin}.
	Note also that due to~\eqref{S2-refl}
	\begin{equation}\label{eta-inv}
		\eta^{-1}(\bm{x}; g) = \eta(\bm{x}; g^*).
	\end{equation}
	For any of the above functions of spatial variables $f(\bx) \equiv f(\bx; g | \bo)$ we define its counterpart depending on spectral variables $\hat{f}(\bl)$ by the formula
	\begin{align}\label{f-hat}
		\hat{f}(\bl) = f(\bl; \hat{g}^*| \hat{\bo}),
	\end{align}
	where we introduced the \textit{dual} parameters
\begin{align*}
	\hat{\bm{\omega}} = \bigl(\omega_2^{-1}, \omega_1^{-1}\bigr), \qquad \hat{g} = \frac{g}{\omega_1 \omega_2}, \qquad \hat{g}^* = \frac{g^*}{\omega_1 \omega_2}.
\end{align*}
For instance,
\begin{align}\label{f0b}
	\hat{\mu}(\l) \equiv \mu(\l; \hat{g}^\ast | \bm{\hat{\omega}}) = S_2(\imath \l | \hat{\bm{\omega}}) \, S^{-1}_2(\imath \l + \hat{g}^\ast| \hat{\bm{\omega}}).
\end{align} 
Note that by assumptions \eqref{g-cond}
\begin{align}
	\Re \hat{g} >0, \qquad \Re \hat{g}^* > 0.
\end{align}
Finally, for a function $w(\bm{s})$ denote the corresponding bilinear form and scalar product (the latter for real nonnegative $w(\bm{s})$)
\begin{align}
	& \bigl(\phi_1(\bm{s}) , \phi_2(\bm{s}) \bigr)_w = \int_{\mathbb{R}^n} d\bm{s} \; w(\bm{s}) \, \phi_1(-\bm{s}) \, \phi_2(\bm{s}), \\[6pt]
	& \bigl\langle \phi_1(\bm{s}) , \phi_2(\bm{s}) \bigr\rangle_w = \int_{\mathbb{R}^n} d\bm{s} \; w(\bm{s}) \, \overline{\phi_1(\bm{s})} \, \phi_2(\bm{s}).
\end{align}

\section{Hamiltonians and wave functions}

\subsection{Symmetries} 

Hamiltonians of Ruijsenaars hyperbolic system have the form
\begin{align}\label{Hh}
	H_s= \sum_{ \substack{J \subset \{1, \dots, n\} \\ |J| = s}} \, \prod_{\substack{j \in J \\ k\not\in J}} \frac{\sh \frac{\pi}{\omega_2} ( x_j - x_k - \imath g)}{ \sh \frac{\pi}{\omega_2} (x_j - x_k) } \, \prod_{j \in J} e^{- \imath \omega_1 \partial_{x_j}},
\end{align}
where $s = 1, \dots, n$. In particular,
\begin{align}
	H_n = e^{- \imath \omega_1 ( \partial_{x_1}+ \ldots + \partial_{x_n})}.
\end{align}
These Hamiltonians commute with each other (see \cite{HR1} and references therein)
\begin{align}
	[H_s, H_r] = 0.
\end{align}
It is convenient to collect them into the generating function
\begin{equation}\label{f1}
	H(\lambda) = e^{- \pi n \omega_1 \lambda} \, H_n^{-1/2} \sum_{s = 0}^n \bigl( e^{2\pi \omega_1 \lambda} \bigr)^{n - s} \, H_s,
\end{equation}
where we additionally denote $H_0 = 1$.

In \cite{HR1} Halln\"as and Ruijsenaars constructed eigenfunctions of the above Hamiltonians. Denote the tuples of $n$ and $n - 1$ variables as 
\begin{align}
	\bl = (\l_1, \dots, \l_n), \qquad \bl' = (\l_1, \dots, \l_{n - 1}).
\end{align}
Then Halln\"as--Ruijsenaars integral representation of the eigenfunctions is given by the following formula
\begin{align}\label{Psi}
	\begin{aligned}
		\Psi_{\bm{\lambda}}(\bm{x}) = d_{n - 1} \int_{\mathbb{R}^{n - 1}} &d\bm{y}' \; \mu(\bm{y}') \; \prod_{j = 1}^n \prod_{k = 1}^{n - 1} K(x_j - y_k)  \\[6pt]
		&\times e^{ 2\pi \imath \, \lambda_n \bigl[(x_1 + \ldots + x_n) - (y_1 + \ldots + y_{n - 1}) \bigr]} \, \Psi_{\bm{\lambda}'}(\bm{y}'),
	\end{aligned}
\end{align}
that is the eigenfunction for $n$ particles is expressed in terms of the eigenfunction for $n - 1$ particles.
The one particle eigenfunction is the plane wave
\begin{align}
	\Psi_{\lambda_1}(x_1) = e^{2\pi \imath \lambda_1 x_1},
\end{align}
and for future convenience we choose normalization constant
\begin{align}\label{dn}
	d_{n - 1} = \bigl[ \sqrt{\omega_1 \omega_2} \, S_2(g| \bm{\omega}) \bigr]^{1 - n}.
\end{align}
The above integral representation is absolutely convergent under restrictions on the parameters \eqref{o-cond}, \eqref{g-cond}, see Proposition \ref{prop:Psi-bound} below.

The wave functions $\Psi_{\bl}(\bx) \equiv \Psi_{\bl}(\bx; g | \bo)$ have many symmetries, some of which easily follow from the integral representation \eqref{Psi}. For example, these functions are clearly symmetric with respect to permutations of spatial variables
\begin{align}
	\bx = (x_1, \dots, x_n) \;\, \mapsto \;\, \sigma(\bx) = (x_{\sigma(1)}, \dots, x_{\sigma(n)}), \qquad \sigma \in S_n.
\end{align}
Besides, Halln\"as--Ruijsenaars representation has some nontrivial symmetries proved in \cite{BDKK2, BDKK4}. Below we list all known symmetries and cite the corresponding statements for the nontrivial ones.
\bigskip

\noindent \textit{1. Bispectral duality \hfill  \cite[Theorem 5]{BDKK2}}
\begin{align}\label{bisp}
	\Psi_{\bm{\lambda}}(\bm{x}; g| \bm{\omega}) = \Psi_{\bm{x}}(\bm{\lambda}; \hat{g}^*| \hat{\bm{\omega}}).
\end{align}

\noindent \textit{2. Symmetry over spatial and spectral variables} 
	\begin{align}
		\Psi_{\bm{\lambda}}(\bm{x}) = \Psi_{\bm{\lambda}}(\sigma(\bm{x})) = \Psi_{\sigma(\bm{\lambda})}(\bm{x}), \qquad \sigma \in S_n.
	\end{align}
	
\noindent \textit{3. Coupling reflection symmetry \hfill  \cite[Theorem 5]{BDKK4}}
\begin{align}\label{refl}
	\Psi_{\bm{\lambda}}(\bm{x};g^*) = \eta(\bm{x}) \, \hat{\eta}(\bm{\lambda}) \, \Psi_{\bm{\lambda}}(\bm{x}; g).
\end{align}

\noindent \textit{4. Modular symmetry}
	\begin{align}
		\Psi_{\bl}(\bx| \omega_1, \omega_2) = \Psi_{\bl}(\bx| \omega_2, \omega_1).
	\end{align}
	
\noindent \textit{5. Parity}  
\begin{align}\label{xl-refl}
	\Psi_{\bm{\lambda}}(\bm{x}) = \Psi_{-\bm{\lambda}}(-\bm{x}).
\end{align}

\noindent \textit{6. Shift relation}
\begin{align}\label{l-shift}
	& \Psi_{\bm{\lambda} + \alpha\bm{e}}(\bm{x}) = e^{2\pi \imath \, \alpha (x_1 + \ldots + x_n)} \, \Psi_{\bm{\lambda}}(\bm{x}), \qquad \bm{e}= (1, \dots, 1).
\end{align}
\vspace{-3pt}

Now let us recall diagonalization property. The eigenvalues of Hamiltonians \eqref{Hh} are expressed in terms of elementary symmetric polynomials 
\begin{align}
	e_s(z_1, \dots, z_n) = \sum_{1 \leq i_1 < \ldots < i_s \leq n} z_{i_1} \cdots z_{i_s}.
\end{align}
Namely, under condition
\begin{align}\label{H-g}
	\Re g < \Re \omega_2
\end{align}
the wave functions satisfy difference equations
\begin{align}\label{HPsi}
	H_s \, \Psi_{\bm{\lambda}}(\bm{x}) = e_s \bigl( e^{ 2\pi \omega_1 \lambda_1}, \ldots, e^{ 2\pi \omega_1 \lambda_n} \bigr) \, \Psi_{\bm{\lambda}}(\bm{x}), \qquad s = 1,\dots n.
\end{align}
Equivalently, we can write these equations using the generating function \eqref{f1}
\begin{align} \label{Hl-Psi}
	H(\lambda) \, \Psi_{\bm{\lambda}}(\bm{x}) = \prod_{j = 1}^n 2\ch [\pi \omega_1 (\lambda - \lambda_j)]  \; \Psi_{\bm{\lambda}}(\bm{x}).
\end{align}
This result is proved in \cite[Proposition 6.2]{HR1} for real periods $\omega_i$, however, its extension to complex valued parameters is immediate, see \cite[Theorem 1]{BDKK2}.

The restriction \eqref{H-g} can be weakened in the following way. Using difference equations for the double sine function \eqref{S2-diff}, one can check that Hamiltonians with reflected coupling constant $g^* = \omega_1 + \omega_2 - g$ differ from the original ones \eqref{Hh} by similarity transformation
\begin{align}\label{f22}
	H(\l; g^*) = \eta(\bm{x}) \; H(\l; g) \; \eta^{-1}(\bm{x}),
\end{align} 
where $\eta(\bx)$ is defined by \eqref{eta}. Due to this relation and the symmetry of wave functions~\eqref{refl} the equation
\begin{align}\label{f21}
	\eta(\bm{x})  H(\l) \, \Psi_{\bm{\lambda}}(\bm{x}) = \prod_{j = 1}^n 2\ch [\pi \omega_1 (\lambda - \lambda_j)]  \; 
	\eta(\bm{x}) \Psi_{\bm{\lambda}}(\bm{x}) 
\end{align}
holds in the region 
\begin{align}\label{f21a}
	\Re g^*<\Re \o_2 \qquad \Leftrightarrow \qquad \Re g > \Re \omega_1.
\end{align}
Here $\eta(\bm{x})  H(\l)$ is regarded as a single operator acting on the wave function. In what follows we use \rf{f21} instead of \rf{Hl-Psi} in the region \rf{f21a}.

\begin{remark}
	Because of bispectral symmetry \eqref{bisp} the Hamiltonians \eqref{Hh} also have dual counterparts that act on the spectral variables $\bl$ and are diagonalized by the same wave functions.
\end{remark}

Finally, Hamiltonians \eqref{Hh} are symmetric with respect to the bilinear form 
\begin{align}\label{bil-pair}
	\bigl(\phi_1(\bx),\phi_2(\bx) \bigr)_{\mu}=\int_{\R^n}d\bx \; \mu(\bx) \, \phi_1(-\bx) \, \phi_2(\bx).
\end{align}
That is we have the equality 
\beq \label{f23}
\int_{\R^n} d\bx \; \mu(\bx) \, \vf_1(-\bx) \, [H(\l) \, \vf_2](\bx)=\int_{\R^n} d\bx \; \mu(\bx) \, [H(\l) \, \vf_1](-\bx) \, \vf_2(\bx)
\eeq
provided the functions $\vf_i(\bx)$ are analytic in the strips
\beq\label{f25}|\Im x_j|\leq  \frac{\Re \o_1}{2}, \qquad j = 1, \dots, n, \eeq
and decay in these strips fast enough, see Lemma~\ref{lemma:Hsym} in Appendix \ref{AppendixB}. 

Due to the relation \rf{f22} symmetry of the Hamiltonians with reflected coupling~$g^*$ can be rewritten in the form
 \beq \label{f24} 
 \bigl( \eta(\bx) \, \vf_1(\bx), H(\l) \, \vf_2(\bx) \bigr)_\Delta= \bigl(\eta(\bx) \, H(\l) \, \vf_1(\bx),\vf_2(\bx) \bigr)_\Delta
 \eeq
assuming analyticity and sufficient decay of the products $\eta(\bx) \, \phi_i(\bx)$ in the strips~\eqref{f25}. Indeed, consider the pairing
	\beq\label{f26} 
		\bigl(\eta(\bx) \, \vf_1(\bx), \, H(\l;g^\ast) \, \eta(\bx) \, \vf_2(\bx)\bigr)_{\mu^*} \, ,
	\eeq
where $\mu^*(\bx) \equiv \mu(\bx; g^*)$. With the help of relations~\eqref{eta-inv}, \eqref{f22} the equality
	\beq	
	\bigl(\eta(\bx) \, \vf_1(\bx), \, H(\l;g^\ast) \, \eta(\bx) \, \vf_2(\bx)\bigr)_{\mu^*} = \bigl( H(\l;g^\ast)\, \eta(\bx) \, \vf_1(\bx), \, \eta(\bx) \, \vf_2(\bx)\bigr)_{\mu^*}
	\eeq
reduces to the identity
\beq \label{f26a}
	\begin{split}
		& \int_{\R^n}d\bx \; \Delta(\bx) \, \eta(\bx) \, \vf_1(-\bx) \, [H(\l; g) \, \vf_2](\bx) \\[6pt]
		& \quad = \int_{\R^n}d\bx \;  \Delta(\bx) \, \eta(\bx) \, [H(\l; g) \, \vf_1 ](-\bx) \, \vf_2(\bx),
	\end{split}
\eeq
which is precisely \rf{f24}. Collecting all together, to the above list of six properties we can add the following two.
\bigskip

\noindent \textit{7. Eigenvalue equations \rf{Hl-Psi} and \rf{f21}.} 
 
\bigskip
\noindent \textit{8. Symmetry of Hamiltonians with respect to the bilinear pairings \rf{f23} and \rf{f26a}.}

\subsection{Bounds}

The bound on the wave functions derived in \cite[Corollary 6]{BDKK2} is important throughout this paper. The following proposition is its slight generalization.
\begin{proposition}\label{prop:Psi-bound}
	Let $\bm{\lambda} \in \mathbb{R}^n$ and $\bm{x} \in \mathbb{C}^n$ such that 
	\begin{align}\label{x-str}
		| \Im x_j | < \frac{\Re g^*}{2}, \qquad j = 1, \dots, n.
	\end{align}
	Then 	for arbitrary $\delta > 0$
	\begin{align}\label{f3}
		\begin{aligned}
			&| \Psi_{\bm{\lambda}}(\bm{x}) |  \leq C(g, \bm{\omega}) \\[3pt]
			& \quad  \times \exp  \biggl( -2\pi \lambda_n \sum_{j = 1}^n \Im x_j + \delta \sum_{j = 1}^n | \Re x_j | - \pi \Re \hat{g} \sum_{1 \leq j < k \leq n} | \Re (x_j - x_k)|  \biggr).
		\end{aligned}
	\end{align}
 
\end{proposition}

\begin{remark}
	Since $\Psi_{\bl}(\bm{x})$ is symmetric in spectral parameters $\bl$, the variable~$\l_n$ in the second line could be replaced by $\min_k\l_k$.
\end{remark}
	
To derive the bound \eqref{f3} from the Halln\"as--Ruijsenaars representation \eqref{Psi} one needs estimates proved in \cite[Appendix A]{BDKK1}. Namely, 
\begin{align}\label{mu-bound}
	& | \mu(x) | \leq C(b, g, \bm{\omega}) \, e^{ \pi \Re \hat{g} \, | \Re x| }, \\[6pt] \label{K-bound}
	& | K(x) | \leq  C(b, g, \bm{\omega}) \, e^{ - \pi \Re \hat{g} \, | \Re x| }
\end{align}
for $x \in \mathbb{C}$ in any closed subset of analyticity domains of these functions and such that $| \Im x | < b$. The rest of arguments are the same, as in the proof of \cite[Corollary~6]{BDKK2}, where only the case $x_j \in \mathbb{R}$ is considered. Note that for $x_j$ from the strips \eqref{x-str} poles of the kernel function $K(x_j - y_k)$ in the integrand \eqref{Psi}
\begin{align}
	y_k = x_j \pm \imath \biggl( \frac{ g^*}{2} +  m_1 \omega_1 +  m_2  \omega_2 \biggr), \qquad m_i \in \mathbb{N}_0
\end{align}
are separated by integration contour $y_k \in \mathbb{R}$.  

\begin{corollary} \label{cor:Psi-an}
	Let $\bm{\lambda} \in \mathbb{R}^n$. The eigenfunction $\Psi_{\bm{\lambda}}(\bm{x})$ is analytic in $x_j \in \mathbb{C}$ in the strips
	\begin{align}\label{x-str2}
		| \Im x_j | < \frac{\Re g^*}{2}, \qquad j = 1, \dots, n.
	\end{align}
\end{corollary}

\begin{proof}
	Due to Proposition \ref{prop:Psi-bound} the integral representation \eqref{Psi} converges uniformly for $x_j$ varying over compact subsets of the strips \eqref{x-str2}. Therefore, it is analytic in $x_j$.
\end{proof}

\begin{remark}
	Domain of analyticity can be easily continued to $\bm{x} \in \mathbb{C}^n$ satisfying
	\begin{align}
		| \Im (x_j - x_k) | < \Re g^*, \qquad j,k = 1, \dots, n,
	\end{align}
	although we won't need it below. To justify this one needs to simultaneously shift straight integration contours for all~$y_k$, see~\cite[Proposition 6.1]{HR1}.
\end{remark}

Let us state two more inequalities. For the half of the measure \eqref{mu-half} from the bound~\eqref{mu-bound} we get
\begin{align}\label{f10}
	| \mu'(\pm\bm{x}) | \leq C(g, \bm{\omega}) \, \exp \biggl( \pi \Re \hat{g} \sum_{1 \leq j < k \leq n} | \Re (x_j - x_k)|\biggr)
\end{align}
provided that $\bx \in \mathbb{C}^n$ is in closed subset of analyticity domain and $\Im \bx$ is bounded. Note also that
\begin{align}
	\eta(\bx) = \prod_{1 \leq j < k \leq n} K(x_j - x_k; \, g^*-g | \bo).
\end{align}
Hence, from the bound \eqref{K-bound} we have
\begin{align} \label{eta-bound}
	| \eta(\bx) | \leq C(g, \bo) \, \exp  \biggl( \pi \Re (\hat{g} - \hat{g}^*) \sum_{1 \leq j < k \leq n} | \Re (x_j - x_k)|\biggr)
\end{align}
under the same assumptions on $\bx$. 

Combining inequality \eqref{f10} with the Proposition~\ref{prop:Psi-bound} and bispectral duality \eqref{bisp} we obtain the following useful estimates.

\begin{corollary}\label{cor:Psi-bound} Let $\bx, \bm{\lambda} \in \mathbb{R}^n$. Then for arbitrary $\delta > 0$
		\begin{align} \label{f11}
			& \bigl| \mu'(\pm\bm{x}) \, \Psi_{\bm{\lambda}}(\bm{x}) \bigr|  \leq C(g, \bm{\omega}) \, \exp  \biggl( \delta \sum_{j = 1}^n | x_j |   \biggr), \\[4pt]
		\label{f11a}
			& \bigl| \hat{\mu}'(\pm\bm{\l}) \, \Psi_{\bm{\lambda}}(\bm{x}) \bigr|  \leq C(g, \bm{\omega}) \, \exp  \biggl( \delta \sum_{j = 1}^n | \l_j |   \biggr).
		\end{align}
\end{corollary}

\section{Spectral transform}
For the functions $\phi(\bx)$ and $\chi(\bl)$ of spatial and spectral variables correspondingly define the integral operators~$T$ and $T^\dagger$ 
\begin{align}\label{T-def}
	& [\T \phi] (\bm{\lambda}) = \int_{\mathbb{R}^n} d\bm{x} \; \mu(\bm{x}) \, \Psi_{\bm{\lambda}}(-\bm{x}) \, \phi(\bm{x}),\\[4pt] \label{f6}
	& [\TT \chi] (\bm{x}) = \int_{\mathbb{R}^n} d\bm{\lambda} \; \hat{\mu}(\bm{\lambda}) \, \Psi_{\bm{\lambda}}(\bm{x}) \, \chi(\bm{\lambda}),
\end{align} 
once the integrals on the right hand sides converge. Equivalently, they can be written as the bilinear pairings
\begin{align}\label{f27} 
	[\T \phi] (\bm{\lambda}) = \bigl(\Psi_{\bm{\lambda}}(\bm{x}), \phi(\bx) \bigr)_{\mu}, \qquad [\TT \chi] (\bm{x}) = \bigl(\Psi_{-\bm{\lambda}}(\bm{x}), \chi(\bl) \bigr)_{\hat{\mu}}.
\end{align}
The defined operators are formally adjoint to each other modulo reflection of variables
\begin{align}\label{TTadj}
	\bigl(  \chi(\bl), \, [T\phi](\bl) \bigr)_{\hat{\mu}}=\bigl( [\TT \chi](-\bx), \, \phi(\bx) \bigr)_{\mu}.
\end{align} 
In the next section we prove that for the suitable space of functions~$\phi(\bx)$ the operator~$T^\dagger$ is left inverse of $T$, see Theorem \ref{theoremf1}. For this we need the composition $T^\dagger T$ to be well defined on those functions $\phi(\bx)$. This will be  justified further in this section, see Corollary \ref{corollaryf3}, which follows from the bound on the image $[T \phi](\bl)$ given by Proposition \ref{propositionf2}.

The bound on the wave function $\Psi_{\bl}(\bx)$~\eqref{f11} implies that the operator $\T$ is well defined on integrable functions $\phi(\bx)$ such that
\begin{align} \label{muh-phi}
	|\mu'(\bx) \, \vf(\bx)| \leq C \exp \biggl(-\ve\sum_{j = 1}^n |x_j|\biggr), \qquad \bx \in \R^n,
\end{align}
with some $\ve > 0$. Due to the estimate for $\mu'(\bx)$~\eqref{f10} the sufficient condition is the bound
 \beq \label{f9a} 	| \phi(\bm{x}) | \leq C \exp \biggl(- \pi \Re \hat{g} \sum_{1 \leq j < k \leq n} |  x_j - x_k|-\ve\sum_{j = 1}^n |x_j|\biggr), \qquad \bx \in \R^n.
 \eeq 
Analogously, the map $\TT$ is well defined on integrable functions $\chi(\bl)$ such that
 \beq \label{f9b} |\hat{\mu}'(\bl) \, \chi(\bl)| \leq C \exp \biggl(-\ve\sum_{j = 1}^n |\l_j|\biggr), \qquad \bl \in \R^n,
 \eeq 
 with some $\ve > 0$.

Now we introduce the space of functions $\S_{\bo,g}$, where as before parameters $\bo, g$ are subjected to conditions~\eqref{o-cond}, \eqref{g-cond}, \eqref{o-ord}. The space $\S_{\bo,g}$ consists of complex valued symmetric functions $\phi(\bx)$ analytic in the strips
\begin{align} \label{f7}
	| \Im x_j | \leq  \frac{\Re {\omega}_1}{2}= \min\left( \frac{\Re \o_1}{2},\,  \frac{\Re \o_2}{2}\right), \qquad j = 1, \dots, n,
\end{align}
and admitting the bound 
\beq \label{f9} 	| \phi(\bm{x}) | \leq C 	 \exp \biggl(- 2\pi \Re \hat{g} \sum_{1 \leq j < k \leq n} | \Re (x_j - x_k)|-\ve\sum_{j = 1}^n |\Re x_j|\biggr)
\eeq 
with some $\ve>0$ in these strips. 

In particular, for such functions due to~\eqref{f10} we have inequality with the full measure function
\beq\label{f8} 
|\mu(\bx)\, \vf(\bx)| \leq C \exp \biggl(-\ve\sum_{j = 1}^n |x_j|\biggr), \qquad \bx \in \R^n.
\eeq
Notice that the last bound \eqref{f8} is stronger than \rf{muh-phi}. Thus, the operator~$\T$ is well defined on the space $\S_{\bo,g}$ and its action is given by absolutely convergent integral. Moreover, due to analiticity in the strips \rf{f7} action of the Hamiltonians' generating functions~$H(\l)$~\eqref{f1} is also well defined on the space $\SC_{\bo,g}$.

To estimate the image $[T\phi](\bl)$ of a function $\phi(\bx) \in \SC_{\bo, g}$ we use Hamiltonians' generating function $H(\l)$ and its symmetry with respect to the bilinear form~\eqref{bil-pair}. For this we need the following lemma.

\begin{lemma}\label{lemmaf1} 
	Let $\phi(\bx) \in \SC_{\bo,g}$. Then the function $\psi(\bx)=H(\l)\, \phi(\bx)$ belongs to the region of definition of the operator $\T$, that is the corresponding integral absolutely converges. Moreover, its image $[\T\psi](\bl)$ admits the bound
	\begin{align}\label{f12}
		\bigl| \hat{\mu}'(\bl) \, [\T\psi](\bl) \bigr| \leq C(g,\bo,\phi) \, \exp\biggl(\delta \sum_{j = 1}^n |\l_j| \biggr), \qquad \bl \in \R^n,
	\end{align}
	with arbitrary small $\delta > 0$. 
\end{lemma}
\begin{proof} 
		By definition, 
		\begin{align}
			[T \psi](\bl) = \int_{\R^n} d\bx \; \mu(\bx) \, \Psi_{\bl}(-\bx) \, H(\l) \, \phi(\bx).
		\end{align}
		The operator $H(\l)$~\eqref{f1} can be represented as the sum
			\begin{align}\label{f18}
			H(\l) = e^{- \pi n \omega_1 \l}  \sum_{J \subset \{1, \ldots, n\}} \bigl( e^{2\pi \omega_1 \l} \bigr)^{n - |J|}  \; {H}_J
		\end{align}
		with the terms
		\begin{align}\label{f19}
			{H}_J = \prod_{\substack{j \in J \\ k\not\in J}} \frac{ \sh \frac{ \pi}{{\omega}_2} ( x_j - x_k - \imath {g} )  }{ \sh \frac{ \pi}{{\omega}_2}  ( x_j - x_k  ) } \prod_{j \in J} e^{ - \frac{\imath {\omega}_1}{2} \partial_{x_j} } \, \prod_{k \not\in J} e^{ \frac{\imath {\omega}_1}{2} \partial_{x_k}  }.
		\end{align}
	 Recall also expression for the measure~\eqref{f00}, \eqref{Skl}
	 \begin{align}
	 	\mu(\bx)= \frac{1}{n!} \prod_{1 \leq j < k \leq n} 4 \sh \frac{ \pi}{{\omega}_1}  ( x_j - x_k  ) \, \sh \frac{ \pi}{{\omega}_2}  ( x_j - x_k  ) \prod_{\substack{j ,k = 1 \\ j \not=k}}^n S^{-1}_2( \imath x_j - \imath x_k +g). 
	 \end{align} 
	Clearly, in the product 
	\beq \mu(\bx) \, {H}_J \, \phi(\bx)\eeq
	the denominators from $H_J$~\eqref{f19} are canceled by corresponding factors from $\mu(\bx)$, so that the resulting coefficient coincides with $\mu(\bx)$ up to some factors $\sh \frac{ \pi}{{\omega}_2}  ( x_j - x_k  )$ replaced by $\sh \frac{ \pi}{{\omega}_2}  ( x_j - x_k -\imath g  )$. This coefficient has no singularities on integration domain $\bx \in \R^n$ and has the same exponential asymptotic behaviour as $\mu(\bx)$ for large $|\bx|$. 
	
	In other words, analogously to \eqref{f8} we have the bound
	\beq\label{f8a} |\mu(\bx) \, H(\l) \, \vf(\bx)| \leq C(g, \bo, \phi) \, \exp \biggl(-\ve\sum_{j = 1}^n |x_j|\biggr), \qquad \bx \in \R^n,
	\eeq
	with some $\ve>0$. On the other hand, we have uniform in $\bx \in \R^n$ estimate for the wave function~\rf{f11a} 
	 \beq 
	 |\hat{\mu}'(\bl) \, \Psi_{\bl}(-\bx) | \leq C(g, \bo) \, \exp \biggl( \delta\sum_{j = 1}^n |\l_j| \biggr) \eeq
	with arbitrary $\delta > 0$. Thus, the integral  
		\beq
		\hat{\mu}'(\bl) \, [\T \psi](\bl)=\hat{\mu}'(\bl)\int_{\R^n}d\bx \; \mu(\bx) \, \Psi_{\bl}(-\bx) \, H(\l) \, \vf(\bx)
		\eeq
		absolutely converges and admits the stated bound~\eqref{f12}.
\end{proof}
		
\begin{proposition}\label{propositionf2} Let $\phi(\bx) \in \SC_{\bo, g}$. Then for $\bl \in \R^n$ its image $[T\vf](\bl)$ admits the bound 
		\begin{align}\label{f13} 
			\bigl| \hat{\mu}'(\bl) \, [T\vf](\bl) \bigr| \leq C(g, \bo, \phi) \, \exp \biggl( [\delta - \pi \Re {\omega}_1] \sum_{j = 1}^n |\l_j| \biggr)
		\end{align}
		with any $\delta>0$.
\end{proposition}

\begin{proof} 
	{\bf 1}.  First, consider the case 
	\begin{align} \label{Reg-as}
		\Re g<\Re \o_2.
	\end{align}
	The action of~$T$ can be written as the bilinear pairing with the wave function~\eqref{f27}
	\begin{align}
		[\T \phi] (\bm{\lambda}) = \bigl(\Psi_{\bm{\lambda}}(\bm{x}), \phi(\bx) \bigr)_\mu.
	\end{align}
	To prove the stated bound~\eqref{f13} we use symmetry of Hamiltonians' generating function~$H(\l)$ with respect to this bilinear pairing (see Lemma~\ref{lemma:Hsym}) and the fact that wave functions diagonalize $H(\l)$~\eqref{Hl-Psi}
	\begin{align}\label{TH-id}
		\begin{aligned}
			[T {H}(\l) \, \phi](\bm{\bl}) &= \bigl( \Psi_{\bl}(\bm{x}) , \, H(\l) \, \phi(\bx) \bigr)_{{\mu}} = \bigl( H(\l) \, \Psi_{\bl}(\bm{x})  , \,  \phi(\bx) \bigr)_{{\mu}} \\[6pt]
			& = \prod_{j = 1}^n 2\ch \bigl( \pi {\omega}_1 (\l - \l_j) \bigr) \; [\T \phi](\bl).
		\end{aligned}
	\end{align}
	Notice that the conditions of Lemma~\ref{lemma:Hsym} are satisfied. The wave function $\Psi_{\bl}(\bx)$ is analytic in the strips
	\begin{align}
		| \Im x_j | \leq  \frac{\Re \omega_1}{2} < \frac{\Re g^*}{2}, \qquad j = 1, \dots, n
	\end{align}
	by Corollary~\ref{cor:Psi-an}, and the same is true for $\phi(\bx) \in \SC_{\bo, g}$. Besides, both functions $\Psi_{\bl}(\bx)$,~$\phi(\bx)$ admit bounds~\eqref{f3},~\eqref{f9}, which suit Lemma~\ref{lemma:Hsym}.

	Using Lemma~\ref{lemmaf1} for the left hand side of identity~\eqref{TH-id} we conclude that
		 \begin{align}\label{muT-b2}
		 	\bigl| \hat{\mu}'(\bl) \, [\T \phi](\bl) \bigr| \leq C(g, \bm{\omega}, \phi) \, \prod_{j = 1}^n \frac{e^{\delta | \l_j|}}{ \bigl| 2\ch \bigl(\pi {\omega}_1 (\l - \l_j)\bigr) \bigr|}
		 \end{align}
		 with arbitrary $\delta > 0$. Next let us estimate the right hand side of the last formula. Clearly, 
		 \begin{align}\label{ch-ineq}
		 	\bigl| 2\ch \bigl(\pi {\omega}_1 (\l - \l_j) \bigr) \bigr| \geq \bigl| 2\sh \bigl(\pi \Re  {\omega}_1 \, (\l - \l_j) \bigr) \bigr|.
		 \end{align}
	 	For the fixed $\bm{\l} \in \mathbb{R}^n$ consider $\l \in \mathbb{R}$ such that
	 \begin{align}
	 	| \l - \l_j | \geq \frac{\ln \sqrt{2}}{\pi \Re {\omega}_1}, \qquad j = 1, \dots, n.
	 \end{align}
	 Then one can write inequalities
	 \begin{align}
	 	\bigl| 2\sh \bigl(\pi \Re {\omega}_1 \, (\l - \l_j) \bigr) \bigr| \geq \frac{1}{2} \, e^{\pi \Re {\o}_1 \, |\l - \l_j|} \geq \frac{1}{2} \, e^{\pi \Re {\o}_1 \, (|\l_j| -  |\l|)}.
	 \end{align}
	 Combining this with~\eqref{muT-b2} and~\eqref{ch-ineq}  we get the estimate
	 \begin{align}
	 	\bigl| \hat{\mu}'(\bm{\l}) \, [T \phi](\bm{x}) \bigr| \leq C' \exp \biggl( [\delta - \pi \Re {\omega}_1] \sum_{j = 1}^n |\l_j| \biggr).
	 \end{align}
	 The choice of $\l$ depends on $\bm{\l}$, however, we need $C'$ to be uniform in~$\bm{\l} \in \mathbb{R}^n$. Surely, it can be made uniform, if we can cover all $\bm{\l} \in \mathbb{R}^n$ using finitely many choices of $\l$. 
	 
	 The choice $\l = 0$ covers all $\bm{\l} \in \mathbb{R}^n$ except the regions where \mbox{$|\l_j | < \ln \sqrt{2}/ \pi \Re \o_1$} for at least one $j$. To cover the unbounded parts of these regions it is enough to consider $n -1$ other choices of $\l$ separated from each other by at least $\ln \sqrt{2}/\pi \Re \o_1$. For example, 
	 \begin{align}
	 	\l^{(k)} = \frac{k}{\pi \Re \o_1}, \qquad k = 0, \ldots, n - 1.
	 \end{align}
	 In total, we cover all $\bm{\l} \in \mathbb{R}^n$ with large enough $|\bm{\l}|$ using finite number of~$\l^{(k)}$. Notice also that it is enough to cover all unbounded parts of $\R^n$, since for compact subsets of~$\R^n$ convergence of the integral $[T \phi](\bl)$ is uniform. Hence, the constant $C'$ can be made uniform for all $\bl \in \R^n$ and this gives us the stated result \eqref{f13}.
	 \smallskip
	 
	  {\bf 2}. Now consider the case
		 \begin{align}
		 	\Re g > \Re \o_1.
		 \end{align}
		Then diagonalization property of wave functions holds in the form~\eqref{f21}
		\begin{align}\label{f21-2}
			\eta(\bx) \, H(\l) \, \Psi_{\bl}(\bx) = \prod_{j = 1}^n 2 \ch [\pi \o_1 (\l - \l_j)] \, \eta(\bx) \, \Psi_{\bl}(\bx),
		\end{align}
		where $\eta(\bx) \, H(\l)$ is regarded as a single operator. Correspondingly, we rewrite the spectral transform~$T$ as the pairing
		\begin{align}
			[\T \phi] (\bm{\lambda}) = \bigl(\eta(\bx) \, \Psi_{\bm{\lambda}}(\bm{x}), \phi(\bx) \bigr)_\Delta.
		\end{align}
		Using symmetry of the operator $H(\l)$ rewritten in the form~\eqref{f24} and the formula~\eqref{f21-2} we arrive at the same result, as in the previous case,
		\begin{align}\label{f30}
			\begin{aligned}
				[T {H}(\l) \phi](\bm{\bl}) &= \bigl( \eta(\bx) \,\Psi_{\bl}(\bm{x}) , \, {H}(\l) \, \phi(\bx) \bigr)_{{\Delta}} = \bigl(\eta(\bx) {H}(\l) \, \Psi_{\bl}(\bm{x}) , \,  \phi(\bx) \bigr)_{{\Delta}} \\[6pt]
				& = \prod_{j = 1}^n 2\ch \bigl( \pi {\omega}_1 (\l - \l_j) \bigr) \; [\T \phi](\bl).
			\end{aligned}
		\end{align}
		The arguments afterwards are the same.
		
		 Let us only remark that to use symmetry of $H(\lambda)$ in the form~\eqref{f24} we need the functions $\eta(\bx) \, \Psi_{\bl}(\bx)$, $\eta(\bx) \, \phi(\bx)$ to be analytic and decaying fast enough in the strips
		\begin{align}\label{x-str4}
			| \Im x_j | \leq \frac{\Re \omega_1}{2} < \frac{\Re g}{2}, \qquad j = 1,\dots, n.
		\end{align}
		The analyticity of $\eta(\bx) \, \Psi_{\bl}(\bx)$ follows from Corollary~\ref{cor:Psi-an} and reflection symmetry~\eqref{refl}, while analyticity of $\eta(\bx) \, \phi(\bx)$ follows from assumption on $\phi(\bx) \in \SC_{\bo, g}$ and the fact that poles of~$\eta(\bx)$ (see \eqref{S2-zeroes})
		\beq x_j-x_k = \imath (g+m_1\o_1+m_2\o_2),\qquad j, k =1, \dots, n, \qquad m_i \in \mathbb{N}_0 \eeq
		are outside of the strips~\eqref{x-str4}. Finally, the sufficient decay follows from the same bounds on $\Psi_{\bl}(\bx)$, $\phi(\bx)$, as before, and the bound on $\eta(\bx)$ given by~\eqref{eta-bound}.
		\smallskip 
		
		{\bf 3}. By our assumption $\Re \o_1 \leq \Re \o_2$. Therefore, it is only left to consider the case
		\beq\label{f28}
			\Re g = \Re \o_1= \Re \o_2.
		\eeq
		Since $\phi \in \SC_{\bo, g}$, the left hand side of the inequality~\rf{f13} is an absolutely convergent integral, just as in all cases above. Moreover, its convergence is uniform in $\bo, g$ around the set~\eqref{f28}.  This follows from the fact that the basic bounds on measure and kernel functions~\eqref{mu-bound}, \eqref{K-bound} are uniform in $\bo, g$ varying on compact subsets of complex plane with restrictions~\eqref{o-cond},~\eqref{g-cond}, see \cite[Appendix A]{BDKK1}. By the above arguments, the function from the right side of inequality~\eqref{f13} is also continuous around the set~\eqref{f28}. Thus, the stated bound extends to the set \eqref{f28} as well.
\end{proof}
	
\begin{corollary}\label{corollaryf3} For any $\vf(\bx)\in\SC_{\bo,g}$ the function 
	\begin{align}
		[\TT\T\vf](\bx) = \int_{\R^n} d\bl \; \hat{\mu}(\bl) \, \Psi_{\bl}(\bx) \, [T \phi](\bl)
	\end{align}
is well defined and is given by absolutely convergent integral.
\end{corollary}

\begin{proof}
	The statement follows from Proposition~\ref{propositionf2} and the condition \rf{f9b}.
\end{proof}

\begin{remark} 
	The decay condition on $\phi(\bx)$ in the above statements can be weakened using the stronger bound on wave functions proved in~\cite[Theorem 2.5]{HR3}
	\begin{align}
		\bigl| \mu'(\bx) \, \hat{\mu}'(\bl) \, \Psi_{\bl}(\bx) \bigr| \leq P(\bx) \, \exp \biggl( -2\pi \sum_{j = 1}^n \lambda_j \Im x_j \biggr),
	\end{align}
	where $P$ is some polynomial and certain conditions on variables $\bx, \bl$ and all parameters are assumed. However, this bound is proved only for real periods $\omega_1, \omega_2$, which restricts its use for now.
\end{remark}

\section{Inversion formula and orthogonality}  

\subsection{Statements} 

In this section we prove the main result of this paper~--- the inversion formula for the spectral transform $T$.  The proof relies on the bounds from the previous sections and the results of the paper~\cite{BDKK3}. As usual, the inversion formula also implies that the operator~$T$ respects the bilinear form.

\begin{theorem}\label{theoremf1} For any function $\vf(\bx)\in\SC_{\bo,g}$ and $\bx \in \R^n$
	\beq \label{inv} 
	[\TT\T\vf](\bx)=\vf(\bx). 
	\eeq
\end{theorem}

\begin{theorem}\label{theoremf2}  
For any functions $\vf_1(\bx), \vf_2(\bx)\in\SC_{\bo,g}$
\beq\label{f14} \bigl([\T\vf_1](\bl), \, [\T\vf_2] (\bl) \bigr)_{\hat{\mu}} = \bigl( \vf_1(-\bx),\vf_2(\bx) \bigr)_\mu. \eeq
\end{theorem}

\begin{remark}
	 The additional sign in the right hand side of the equality~\eqref{f14} comes from the fact that $T$ and $T^\dagger$ are adjoint modulo reflection, see~\eqref{TTadj}. It is in accordance with widely used specialization of convolution property in the theory of Laplace and Mellin transforms, for example see~\cite[(2.4.8)]{AAR}.
\end{remark}

 The statements of Theorems~\ref{theoremf1},~\ref{theoremf2} admit reformulation in a language of distributions. Namely, in $(\SC_{\bo,g})'$ we have the equality
\beq \label{f15}
\bigl( \Psi_{\bl}(\by),\Psi_{\bl}(\bx) \bigr)_{\hat{\mu}}=\mu^{-1}(\bx) \, \delta(\bx,\by).\eeq
Here we consider $\SC_{\bo,g}$ with a topology of convergence on compact subsets inside the strips~\rf{f7} and the space $(\SC_{\bo,g})'$ of continuos functionals with corresponding weak topology. The inclusion of integrable functions $\a(\bx)$ into  $(\SC_{\bo,g})'$ is given by the relation
\beq 
\bigl(\alpha(\bx), \vf(\bx) \bigr) =\int_{\R^n} d\bx \; \mu(\bx) \, \alpha(\bx) \, \vf(\bx). \eeq
By $\delta(\bx, \by)$ we denote symmetrized product of delta functions
\begin{align}
	\delta(\bx, \by) = \frac{1}{n!} \sum_{\sigma \in S_n} \delta \bigl( x_1 - y_{\sigma(1)}\bigr) \cdots \delta \bigl( x_n - y_{\sigma(n)}\bigr),
\end{align}
which corresponds to the identity functional for symmetric functions.

Due to bispectral duality~\eqref{bisp}
\begin{align}
	\Psi_{\bm{\lambda}}(\bm{x}; g| \bm{\omega}) = \Psi_{\bm{x}}(\bm{\lambda}; \hat{g}^*| \hat{\bm{\omega}})
\end{align}
we have the same statements with spatial and spectral variables interchanged. That is, for any functions $\vf_1(\bl), \vf_2(\bl)\in\SC_{\hat{\bo},\hat{g}^\ast}$
\begin{align} \label{f14a} 
	[T \, T^\dagger \phi_i] (\bl) = \phi_i(\bl), \qquad  \bigl([ \TT\vf_1] (\bx), \, [\TT\vf_2](\bx) \bigr)_{\mu} = \bigl( \vf_1(-\bl),\vf_2(\bl) \bigr)_{\hat{\mu}},
\end{align} 
or equivalently, in the space of functionals $(\SC_{\hat{\bo},\hat{g}^\ast})'$
\beq \label{f15a} 
\bigl( \Psi_{\bl}(\bx),\Psi_{\bm{\rho}}(\bx) \bigr)_{\mu}=\hat{\mu}^{-1}(\bl) \, \delta(\bl,\bm{\rho}).\eeq
The last formula represents orthogonality of wave functions, while the previous one~\eqref{f15} represents their completeness.

\subsection{Delta sequence}

To prove the above theorems let us recall delta sequence constructed in \cite{BDKK3}. 
Introduce the regularizing function 
\begin{align}\label{R}
	R_{\l, \epsilon}(\bm{\l})  = \exp \biggl( \pi  ( {g}^\ast - 2 \varepsilon ) \biggl[ n\l -\sum_{j = 1}^n \l_j \biggr] \biggr) \, \prod_{j = 1}^n \hat{K}(\l - \l_j)
\end{align}
with regularization parameters $\lambda, \epsilon > 0$. Here, by definition~\eqref{f-hat},
\begin{align}
	\hat{K}(\lambda) \equiv K(\lambda; \hat{g}^* | \hat{\bo}) = S_2^{-1} \biggl(\imath \lambda +\frac{\hat{g}}{2} \bigg| \hat{\bo} \! \biggr) \, S_2^{-1} \biggl( - \imath \lambda +\frac{\hat{g}}{2} \bigg| \hat{\bo} \! \biggr).
\end{align}
From the bound on $K$-function~\eqref{K-bound} we have the estimate
\begin{align}\label{Rba}
	&| R_{\l, \epsilon}(\bm{\l})| \leq C(g,\bo) \, \exp \biggl( -2\pi \epsilon \sum_{j = 1}^n | \lambda - \lambda_j| \biggr), \qquad \bl \in \mathbb{R}^n
\end{align}
provided that $0 \leq \epsilon \leq \Re g^*/2$. Due to this estimate and asymptotics of the double sine function~\eqref{S2-asymp} the regularizing function has two important properties
\begin{align}\label{Rb}
	& | R_{\l, \epsilon}(\bl) | \leq C(g, \bm{\omega}), \qquad\quad  \lim_{\l \to \infty} \, \lim_{\varepsilon \to 0^+} R_{\l, \epsilon}(\bl)  = 1.
\end{align}
Next consider the regularized bilinear pairing between wave functions
\begin{align}\label{reg-pair}
		\bigl( \Psi_{\bm{\l}}(\by) ,\, \Psi_{\bm{\lambda}}(\bx) \bigr)_{\hat{\mu}}^{\l, \varepsilon} & = \int_{\mathbb{R}^n} d\bm{\l} \; \hat{\mu}(\bm{\l}) \, \Psi_{-\bm{\l}}(\bm{y}) \, \Psi_{\bm{\lambda}}(\bm{x})  \, R_{\l, \epsilon}(\bm{\l}),
\end{align}
where we assume $\bx, \by \in \mathbb{R}^n$. In \cite[Proposition 1]{BDKK3} it was proved that this integral is absolutely convergent and can be calculated explicitly
\begin{align} \label{f31} 
	\begin{aligned}
		\bigl( \Psi_{\bm{\l}}(\by) ,\, \Psi_{\bm{\lambda}}(\bx) \bigr)_{\hat{\mu}}^{\l, \varepsilon} & = \bigl[ \sqrt{\omega_1 \omega_2} \, S_2(\hat{g} | \hat{\bo}) \bigr]^{-n} \; \exp \biggl( 2\pi \imath \, \lambda \sum_{j = 1}^n(x_j - y_j) \biggr) \\[6pt]
		& \times \prod_{j, k = 1}^n K \biggl( x_j - y_k + \frac{\imath g^*}{2} - \imath\ve \biggr)
	\end{aligned}
\end{align} 
Moreover, by \cite[Proposition 2]{BDKK3} it forms a delta sequence in the Schwartz space
	\begin{align}\label{f32}
	\lim_{\lambda \to \infty} \, \lim_{\varepsilon \to 0^+} \bigl( \Psi_{\bm{\l}}(\by),\, \Psi_{\bm{\lambda}}(\bx) \bigr)_{\hat{\mu}}^{\l, \varepsilon} = \mu^{-1}(\bm{x}) \, \delta(\bm{x}, \bm{y}).
\end{align}
The last formula is the key result needed for the proof of Theorem~\ref{theoremf1}.

\subsection{Proof of Theorem~\ref{theoremf1}}

	 The product $[\TT\T \vf](\bx)$ is given by the integral 
	\begin{align}
		[T^\dagger T \phi ] (\bm{x}) = \int_{\mathbb{R}^n} d\bm{\l} \; \hat{\mu}(\bm{\l}) \, \Psi_{\bm{\l}}(\bm{x}) \, [\T \phi](\bm{\l}) 
	\end{align}
	with the test function $\phi(\bx) \in\mathcal{S}_{\bo, g}$ and $\bm{x} \in \mathbb{R}^n$. According to Corollary~\ref{corollaryf3} this integral is absolutely convergent. The regularization function $R_{\l, \epsilon}(\bm{\l})$~\eqref{R} is bounded uniformly in $\l, \epsilon, \bm{\l}$, see \eqref{Rb}, so we can insert it inside the above integral and use dominated convergence theorem to interchange limits and integration
	\begin{align} \nonumber
			[\TT \T \phi ] (\bm{x}) &= \int_{\mathbb{R}^n} d\bm{\l} \; \hat{\mu}(\bm{\l}) \, \Psi_{\bm{\l}}(\bm{x}) \, [\T \phi](\bm{\l}) \, \lim_{\l \to \infty} \, \lim_{\epsilon \to 0^+} R_{\l, \epsilon}(\bm{\l}) \\[6pt]  \label{f33}
			& =  \lim_{\l \to \infty} \,\lim_{\epsilon \to 0^+} \int_{\mathbb{R}^n} d\bm{\l} \; \hat{\mu}(\bm{\l}) \, \Psi_{\bm{\l}}(\bm{x}) \, [T \phi](\bm{\l}) \, R_{\l, \epsilon}(\bm{\l}).
	\end{align}
	Moreover, in the integral \rf{f33} for fixed $\lambda, \epsilon$ we can reverse the ordering of integrals to obtain regularized pairing of wave functions \eqref{reg-pair}
	\begin{align} \nonumber
			&\int_{\mathbb{R}^n} d\bm{\l} \; \hat{\mu}(\bm{\l}) \, \Psi_{\bm{\l}}(\bm{x}) \, [\T \phi](\bm{\l}) \, R_{\l, \epsilon}(\bm{\l}) \\[6pt] \nonumber
			&\qquad = \int_{\mathbb{R}^n} d\bm{\l} \; \hat{\mu}(\bm{\l}) \, \Psi_{\bm{\l}}(\bm{x})\, R_{\l, \epsilon}(\bm{\l}) \int_{\mathbb{R}^n} d \bm{y} \; {\mu}(\bm{y}) \, \Psi_{\bm{\lambda}}(-\bm{y}) \, \phi(\bm{y}), \\[6pt] \label{int-regpair}
			&\qquad = \int_{\mathbb{R}^n} d\bm{y} \; {\mu}(\bm{y}) \, \phi(\bm{y}) \, \bigl( \Psi_{\bm{\l}}(\by) ,\, \Psi_{\bm{\lambda}}(\bx) \bigr)_{\hat{\mu}}^{\l, \varepsilon}.
	\end{align}
	This is justified by Fubini's theorem, since the whole integral over $\bm{\lambda}, \bm{y}$ is absolutely convergent. Indeed, applying bound \rf{f11a} for the product $\hat{\mu}(\bl) \, \Psi_{\bl}(\bx) \, \Psi_{\bl}(-\by)$, bound \rf{f8} for the product $\mu(\by) \, \vf(\by)$ and the bound \rf{Rba} for $ R_{\l, \epsilon}(\bm{\l})$ we see that the integrand 
	\begin{align}
		\hat{\mu}(\bl) \, \Psi_{\bl}(\bx) \, \Psi_{\bl}(-\by) \, \mu(\by) \, \phi(\by) \, R_{\l, \epsilon}(\bl)
	\end{align} 
	is bounded by
	\begin{align}
		C(g,\bo, \l, \ve) \, \exp \biggl(
		-\ve'\sum_{j = 1}^n |y_j| - (2\pi\ve-\delta)\sum_{j = 1}^n |\l_j| \biggr) 
	\end{align}
	with some $\ve'>0$ dictated by $\vf(\bx)$ and arbitrary $\delta>0$. Thus, the integral over $\bl, \by$ is absolutely convergent.
	
	The regularized pairing in~\eqref{int-regpair} is a delta-sequence~\eqref{f32}. So, collecting all together we obtain the inversion formula
	\begin{align}
		\begin{aligned}
			[\TT\T \phi ] (\bm{x}) &= \lim_{\l \to \infty} \,\lim_{\epsilon \to 0^+} \int_{\mathbb{R}^n} d\bm{y} \; {\mu}(\bm{y}) \, \phi(\bm{y}) \,
			\bigl( \Psi_{\bm{\l}}(\by) ,\, \Psi_{\bm{\lambda}}(\bx) \bigr)_{\hat{\mu}}^{\l, \varepsilon}
			 \\[6pt]
			& = \int_{\mathbb{R}^n} d\bm{y} \; \phi(\bm{y}) \, \delta(\bm{x}, \bm{y}) = \phi(\bm{x}),
		\end{aligned}
	\end{align}
	where in the last step we use symmetry of $\phi(\bm{y})$ with respect to $y_j$. \hfill{$\Box$}
 
\subsection{Proof of Theorem \ref{theoremf2}}

The statement essentially follows from inversion formula applied to $\phi_2(\bx) \in \SC_{\bo, g}$ 
\begin{align}
	\phi_2(\bx) = \int_{\mathbb{R}^n} d\bl \; \hat{\mu}(\bl) \, \Psi_{\bl}(\bx) \, [T \phi_2] (\bl).
\end{align}
Let us multiply it by $\mu(\bx) \, \phi_1(\bx)$, where $\phi_1(\bx) \in \SC_{\bo, g}$, and integrate over $\bx$
\begin{align}
	\bigl( \phi_1(-\bx), \, \phi_2(\bx) \bigr)_\mu = \int_{\mathbb{R}^n} d\bx \; \mu(\bx) \, \phi_1(\bx) \, \int_{\mathbb{R}^n} d\bl \; \hat{\mu}(\bl) \, \Psi_{\bl}(\bx) \, [T \phi_2] (\bl).
\end{align}
The whole integral over $\bl, \bx$ is absolutely convergent due to Corollary~\ref{cor:Psi-bound}, Proposition~\ref{propositionf2} and since $\phi_1(\bx) \in \SC_{\bo, g}$ admits the bound~\eqref{f8}. Therefore, we can integrate over $\bx$ first
\begin{align}
	\begin{aligned}
		\bigl( \phi_1(-\bx), \, \phi_2(\bx) \bigr)_\mu & = \int_{\mathbb{R}^n} d\bl \; \hat{\mu}(\bl) \, [T \phi_2](\bl) \, \int_{\mathbb{R}^n} d\bx \; \mu(\bx) \, \Psi_{\bl}(\bx) \, \phi_1(\bx) \\[6pt]
		& = \int_{\mathbb{R}^n} d\bl \; \hat{\mu}(\bl) \, [T \phi_2](\bl) \, [T \phi_1](-\bl) = \bigl( [T \phi_1] (\bl), \, [T \phi_2](\bl) \bigr)_{\hat{\mu}},
	\end{aligned}
\end{align}
where we used parity of wave functions $\Psi_{\bl}(\bx) = \Psi_{-\bl}(-\bx)$. \hfill{$\Box$}

 Let us in addition demonstrate that one can prove Theorem~\ref{theoremf2} independently from inversion formula, but using the same delta sequence. Namely, for the functions $\phi_1(\bx), \phi_2(\bx)$ from $\SC_{\bo, g}$ the bilinear pairing 
 \beq  \label{f34} \bigl([\T\vf_1](\bl),\,[\T\vf_2](\bl)\bigr)_{\hat{\mu}} = \int_{\R^n} d\bl\;  \hat{\mu}(\bl) \, [\T\vf_1](-\bl) \, [\T\vf_2](\bl)\eeq
 is given by absolutely convergent integral due to the bound~\eqref{f13}. So, we can replace it by the double limit
  \beq \label{f35} \bigl([\T\vf_1](\bl),\,[\T\vf_2](\bl)\bigr)_{\hat{\mu}} = \lim_{\l \to \infty} \,\lim_{\epsilon \to 0^+} \int_{\mathbb{R}^n} d\bm{\l} \; \hat{\mu}(\bm{\l}) \, [\T\vf_1](-\bl) \, [\T\vf_2](\bl) \, R_{\l, \epsilon}(\bm{\l}). \eeq
  For fixed $\l, \ve$ the whole triple integral~\eqref{f35}
  \begin{align}
  	\int_{\R^{3n}}d\bl\, d\bx\, d\by \; \hat{\mu}(\bl) \, \mu(\bx) \, \mu(\by) \, \Psi_{-\bl}(-\bx) \, \Psi_{\bl}(-\by) \, \vf_1(\bx) \, \vf_2(\by) \, R_{\l, \epsilon}(\bm{\l})
  \end{align}  
  absolutely converges by the same arguments, as before. Hence, by Fubini's theorem we can integrate over $\bl$ first, so that
\begin{align} \nonumber
		& \bigl([\T\vf_1](\bl),\,[\T\vf_2](\bl)\bigr)_{\hat{\mu}} = \lim_{\l \to \infty} \,\lim_{\epsilon \to 0^+} \int_{\R^{2n}} d\bx\, d\by\; \mu(\bx) \, \mu(\by) \, \vf_1(\bx) \, \vf_2(\by) \, 	\bigl( \Psi_{\bm{\l}}(\by) ,\, \Psi_{\bm{\lambda}}(\bx) \bigr)_{\hat{\mu}}^{\l, \varepsilon} \\[6pt] 
		& \quad = \int_{\R^{2n}} d\bx\, d\by\; \mu(\bx) \, \vf_1(\bx) \, \vf_2(\by) \, \delta(\bm{x}, \bm{y}) = \bigl( \vf_1(-\bx),\vf_2(\bx) \bigr)_{\mu},
\end{align} 
where we used the delta sequence~\eqref{f32} and symmetry of $\phi(\by)$ with respect to $y_j$. 

\section{Unitarity regimes} \label{sec:un}

\subsection{Restrictions on parameters and scalar products}

We know four regimes of unitarity of Ruijsenaars hyperbolic system.
  Namely, we have two options for the periods
 \begin{align}\label{omega-un}
 	\omega_1, \omega_2 \in \mathbb{R}, \qquad \text{or} \qquad \bar{\omega}_1 = \omega_2,
 \end{align}
 and two (independent) options for the coupling
 \begin{align}
 	g \in \mathbb{R}, \qquad \text{or} \qquad \bar{g} = g^*.
 \end{align}
Here and in what follows bars mean complex conjugation.
Denote four possible combinations by roman numerals
\begin{align*} 
	&\mathrm{I} \colon && \hspace{-2.5cm} \o_1,\o_2 \in \R, && \hspace{-2cm} g\in\R,\\[2pt]
	&\mathrm{II} \colon && \hspace{-2.5cm} \bar{\omega}_1 = \omega_2,&&  \hspace{-2cm} g\in\R,\\[2pt]
	&\mathrm{III} \colon && \hspace{-2.5cm} \o_1, {\o}_2\in\R,&& \hspace{-2cm} \bar{g}=g^\ast,\\[2pt]
	&\mathrm{IV} \colon && \hspace{-2.5cm} \bar{\omega}_1 = \omega_2,&& \hspace{-2cm} \bar{g}=g^\ast.
\end{align*}
Notice that in all cases we have real $\omega_1 \omega_2$ and $\omega_1 + \omega_2$. Moreover, as before, we always assume conditions on parameters~\eqref{o-cond},~\eqref{g-cond}, which in the above regimes reduce to
\begin{align}
	\Re \omega_1 > 0, \qquad \Re \omega_2 > 0, \qquad 0 < \Re g < \omega_1 + \omega_2.
\end{align}
The condition $\bar{g} = g^* \equiv \omega_1 + \omega_2 - g$ is equivalent to fixing coupling's real part 
\begin{align}
	g = \frac{\omega_1 + \omega_2}{2} + \imath \Im g.
\end{align}

\begin{remark}
	The first regime is known due to the work of Ruijsenaars~\cite[Section~2.2]{R3}. The third regime in the context of Rujisenaars model is considered in~\cite{DKKSS} and has previously appeared in modular $XXZ$ spin chain~\cite{DKM1}. Unitarity with complex conjugated periods in the context of modular double of $U_q(sl_2)$ is discussed in \cite{F1, F2, KLS, BT}. 
\end{remark} 

Now let us describe differences between unitarity regimes. For this recall that wave functions $\Psi_{\bl}(\bx| \omega_1, \omega_2)$ diagonalize the Hamiltonians~\eqref{Hh}
\begin{align} \label{Hs}
	H_s(\omega_1, \omega_2)= \sum_{ \substack{J \subset \{1, \dots, n\} \\ |J| = s}} \, \prod_{\substack{j \in J \\ k\not\in J}} \frac{\sh \frac{\pi}{\omega_2} ( x_j - x_k - \imath g)}{ \sh \frac{\pi}{\omega_2} (x_j - x_k) } \, \prod_{j \in J} e^{- \imath \omega_1 \partial_{x_j}}, \qquad s = 1, \dots, n,
\end{align}
where we emphasized in notation the dependence on periods. Since wave functions are symmetric in $\omega_1, \omega_2$, they also diagonalize operators with interchanged periods~$H_s(\omega_2, \omega_1)$. The regimes I, II differ from the pair III, IV by scalar products, with respect to which all these Hamiltonians are symmetric. 

Namely, consider regimes I, II and the scalar product
 \beq\label{f20} \langle \phi_1(\bx), \phi_2(\bx) \rangle_\mu =\int_{\R^n} d\bx \; \mu(\bx) \, \overline{\phi_1(\bx)} \, \phi_2(\bx)
 \eeq
 with the same measure $\mu(\bx)$, as in the bilinear pairing~\eqref{bil-pair} used in previous sections,
 	\begin{align}\label{mu}
 	\mu(\bx) = \frac{1}{n!} \prod_{\substack{j ,k = 1 \\ j \not=k}} \mu(x_j - x_k), \qquad \mu(x) = \frac{S_2(\imath x | \bo)}{S_2(\imath x + g | \bo)}.
 \end{align}
Note that in these cases $\mu(\bx)\geq 0$. Then in regime I (real periods) all Hamiltonians $H_s(\omega_1, \omega_2)$ and $H_s(\omega_2, \omega_1)$ are symmetric with respect to this scalar product
\begin{align}\label{f40a}
	H_s^\dagger(\omega_1, \omega_2) = H_s(\omega_1, \omega_2), \qquad H_s^\dagger(\omega_2, \omega_1) = H_s(\omega_2, \omega_1),
\end{align}
 while in regime II (complex conjugated periods) we have the equality
 \begin{align}\label{f40} 
 	H^\dagger_s(\o_1,\o_2)=H_s(\o_2,\o_1),
 \end{align}
 so that the symmetric operators are the linear combinations
 \beq\label{f41}  H_s(\o_1,\o_2)+ H_s(\o_2,\o_1), \qquad \imath \bigl[ H_s(\o_1,\o_2)- H_s(\o_2,\o_1) \bigr].\eeq
 This follows from the difference equations on the measure~$\mu(\bx)$, the details are given in Appendix~\ref{AppendixB}.
 
 For regimes III, IV consider scalar product with the different measure $\Delta(\bx)$
 \begin{align}\label{f42}
 	\langle \phi_1(\bx), \phi_2(\bx) \rangle_\Delta =\int_{\R^n} d\bx \; \Delta(\bx) \, \overline{\phi_1(\bx)} \, \phi_2(\bx)
 \end{align}
given by the formula~\eqref{Skl}
 \begin{align}
 	\Delta(\bm{x}) = \frac{\mu(\bm{x}) }{ \eta(\bm{x})} = \frac{1}{n!} \prod_{1 \leq j < k \leq n} 4  \sh \frac{ \pi (x_j - x_k)}{{\omega}_1}\sh \frac{ \pi(x_j - x_k)}{{\omega}_2} \geq 0.
 \end{align}
 Then in regime III all Hamiltonians $H_s(\omega_1, \omega_2)$ and $H_s(\omega_2, \omega_1)$ are symmetric with respect to the latter scalar product, whereas in regime IV we have the relation~\rf{f40}, so that the symmetric operators are again given by the combinations~\rf{f41}, see Lemma~\ref{lemma:Hsym2} in Appendix~\ref{AppendixB}.

 Below we consider spectral transforms defined by means of the above scalar products. These transforms essentially coincide with the previously studied transform~$T$, defined by bilinear pairing~\eqref{f27}. Consequently, these transforms admit the same inversion formula, which can be used to extend them to isomorphisms of the corresponding $L_2$ spaces. This gives a unified proof of completeness and orthogonality of wave functions in all unitarity regimes.

\subsection{Regimes I, II}

 Consider any of the two options for the periods
 \begin{align}
 	\omega_1, \omega_2 > 0, \qquad \text{or} \qquad  \bar{\omega}_1 = \omega_2, \quad \Re \omega_1 > 0, \quad \Re \omega_2 > 0
 \end{align}
 and real coupling 
 \begin{align}
 	0 < g < \omega_1 + \omega_2.
 \end{align}
 Then due to explicit formulas for the wave functions~\eqref{Psi} and the measure~\eqref{mu} we have
 \begin{align}\label{conj}
 	\Psi_{\bm{\lambda}}(-\bm{x}) = \overline{\Psi_{\bm{\lambda}}(\bm{x})}, \qquad \mu(\bm{x}) \geq 0, \qquad \hat{\mu}(\bm{\lambda}) \geq 0, \qquad \bx, \bl \in \mathbb{R}^n.
 \end{align}
As a consequence, the transforms $F$ and $F^\dagger$ defined by means of the scalar product~\eqref{f20} and the one with dual measure
\begin{align} \label{f48} 
	\begin{aligned}
		& [F\phi](\bl) = \langle \Psi_{\bl}(\bx),\phi(\bx)\rangle_\mu,\\[6pt]
		& [F^\dagger\chi](\bx) = \langle\overline{\Psi_{\bl}(\bx)},\chi(\bl)\rangle_{\hat{\mu}}
	\end{aligned}
\end{align}
coincide with $\T$ and $\TT$~\eqref{f27}
\begin{align}
	F = T, \qquad F^\dagger = T^\dagger.
\end{align}
So, the inversion formula~\eqref{inv} remains valid
\begin{align}\label{invF}
	[F^\dagger F \phi](\bx) = \phi(\bx),
\end{align}
whereas the equivariance property~\eqref{f14} turns into 
\beq\label{f49} \bigl\langle[F\vf_1](\bl),\,[F\vf_2](\bl) \bigr\rangle_{\hat{\mu}} = \langle\vf_1(\bx), \, \vf_2(\bx)\rangle_\mu, \eeq
where $\phi_1, \phi_2 \in \SC_{\bo, g}$.  The arguments for the last formula are the same, as for the Theorem~\ref{theoremf2}
\beq
	\bigl([\T\vf_1](\bl), \, [\T\vf_2] (\bl) \bigr)_{\hat{\mu}} = \bigl( \vf_1(-\bx),\vf_2(\bx) \bigr)_\mu.
\eeq
The difference in the sign of $\bx$ from the right arises, because
\begin{align}
	& \bigl\langle \chi(\bl), \, [F \phi](\bl) \bigr\rangle_{\hat{\mu}} = \bigl\langle [F^\dagger \chi](\bx), \, \phi(\bx) \bigr\rangle_{\mu}, \\[8pt]
	& \bigl( \chi(\bl), \, [F \phi](\bl) \bigr)_{\hat{\mu}} = \bigl( [F^\dagger \chi](-\bx), \, \phi(\bx) \bigr)_{\mu},
\end{align}
that is~$F^\dagger$ is adjoint to $F$ with respect to the scalar product and at the same time it is adjoint \textit{modulo reflection} with respect to the bilinear pairing.

	In particular, the equality \eqref{f49} holds for the subset of functions
 	\begin{align} \label{Pset}
 		\mathcal{P} =  \bigl\{ e^{- \bm{x}^2}  P(\bm{x}) \, \big| \, \text{$P$ --- symmetric polynomial} \bigr\} \subset \S_{\bo, g}.
 	\end{align}
 	This subset is dense in the space $L_2^{\mathrm{sym}}(\mathbb{R}^n, {\mu})$, which consists of symmetric functions square integrable with respect to the scalar product~\eqref{f20}. Recall classical arguments for that~\cite[VIII.4.3]{KF}. Let $h(\bx)$ be from orthogonal complement of~$\mathcal{P}$ in $L_2^{\mathrm{sym}}(\mathbb{R}^n, {\mu})$, that~is
   \begin{align}\label{h-as}
   	\int_{\mathbb{R}^n} d\bx \; \mu(\bx) \, \overline{\phi(\bx)} \, h(\bx) = 0
   \end{align}
   for all $\phi(\bx) \in\mathcal{P}$. Since the measure $\mu(\bx)$ has at most exponential growth~\eqref{f10}, the function
   \begin{equation}
   		f(\bx)=\mu(\bx) \, h(\bx) \, e^{- \bm{x}^2}
   \end{equation}
   	is absolutely integrable, decreases more than exponentially, and its Fourier image
   	\begin{align}
   		\check{f}(\bl) = \int_{\mathbb{R}^n} d\bx \; e^{\imath \bl \cdot \bx} \, \mu(\bx) \, h(\bx) \, e^{-\bx^2}
   	\end{align}
   	 is symmetric function analytic in $\C^n$. The Taylor coefficients of $\check{f}(\bl)$ at the point $\bl = 0$ are given by scalar products of $h(\bx)$ with elements from $\mathcal{P}$, so they equal to zero by assumption~\eqref{h-as}. Thus, $\check{f} \equiv 0$, and from Plancherel theorem
   	 \begin{align}
   	 	\int_{\mathbb{R}^n} d\bx \; |f(\bx)|^2 = \frac{1}{(2\pi)^n} \, \int_{\mathbb{R}^n} d\bl \; |\check{f}(\bl) |^2 = 0 
   	 \end{align}
   	 we have $f \equiv 0$, so that $h \equiv 0$. This proves that $\mathcal{P}$ is dense.
   	
   	Because of Proposition~\ref{propositionf2} the images~$[F \phi_i](\bl)$ in the identity~\eqref{f49} belong to the space~$L_2^{\mathrm{sym}}(\mathbb{R}^n, \hat{\mu})$. Hence, by a standard procedure~\cite[VIII.5]{KF} the transform $F$ extends from the dense subset $\mathcal{P}$ to the whole $L_2$ space
   	\begin{align}
   		F \colon \; L_2^{\mathrm{sym}}(\mathbb{R}^n, {\mu}) \; \to \; L_2^{\mathrm{sym}}(\mathbb{R}^n, \hat{\mu}).
   	\end{align}
   	In other words, $F$ is an isometry and the identity
   	\begin{align}
   		F^\dagger F = \mathrm{Id}
   	\end{align}
   	holds in $L_2$ sense. Due to the bispectral duality~\eqref{bisp}
   	\begin{align}\label{bisp2}
   		\Psi_{\bm{\lambda}}(\bm{x}; g| \bm{\omega}) = \Psi_{\bm{x}}(\bm{\lambda}; \hat{g}^*| \hat{\bm{\omega}})
   	\end{align}
   	the same identity holds in reverse order
   	\begin{align}
   		F F^\dagger = \mathrm{Id}.
   	\end{align}
   	Thus, $F$ is unitary operator. 
 
\subsection{Regimes III, IV}

 Now we have $\bar{g} = g^*$, or equivalently
 \begin{align}
 	g = \frac{\omega_1 + \omega_2}{2} + \imath \Im g, \qquad \Im g \in \mathbb{R}.
 \end{align}
Due to reflection symmetry \eqref{refl} and explicit formula~\eqref{Psi}
 \begin{align}\label{conj2}
 		\eta(\bx) \, \hat{\eta}(\bl) \, \Psi_{\bm{\lambda}}(-\bm{x}; g) & = \Psi_{\bm{\lambda}}(-\bm{x}; g^*) =  \overline{\Psi_{\bm{\lambda}}(\bm{x}; g)},
 \end{align}
 where we assume $\bm{x}, \bm{\lambda} \in \mathbb{R}^n$. Under the same assumption
 \begin{align}\label{f51}
 	& \Delta(\bm{x}) = \frac{\mu(\bm{x}) }{ \eta(\bm{x})} \geq 0, \qquad  \hat{\Delta}(\bm{\lambda}) = \frac{\hat{\mu}(\bm{\lambda})}{ \hat{\eta}(\bm{\lambda}) }  \geq 0, \qquad |\eta(\bm{x})|=|\hat{\eta}(\bm{\lambda})|=1,
 \end{align} 
 where the last two equalities follow from definition of $\eta(\bx)$~\eqref{eta} and reflection equation for the double sine function~\eqref{S2-refl}.
 
From~\eqref{conj2} we obtain the following relations between the scalar products and bilinear pairings
\begin{align}\label{f53}
	\begin{aligned}
		\langle\Psi_{\bl}(\bx),\vf(\bx)\rangle_\Delta& = \hat{\eta}(\bl) \, \bigl( \Psi_{\bl}(\bx), \, \vf(\bx) \bigr)_\mu, \\[6pt] 
		\langle \overline{\Psi_{\bl}(\bx)},\chi(\bl)\rangle_{\hat{\Delta}}& = \bigl( \Psi_{-\bl}(\bx), \, \hat{\eta}^{-1}(\bl) \, \chi(\bl) \bigr)_{\hat{\mu}}.
	\end{aligned}
\end{align}
These identities imply that the transforms
\begin{align}\label{f54} 
	\begin{aligned} 
		& [F\phi](\bl)= \langle \Psi_{\bl}(\bx),\phi(\bx)\rangle_\Delta,\\[6pt]
		& [F^\dagger\chi](\bx)=\langle\overline{\Psi_{\bl}(\bx)},\chi(\bl)\rangle_{\hat{\Delta}}
	\end{aligned}
\end{align}
are related to the transforms $\T$ and $\TT$~\eqref{f27} by the formulas
\beq\label{f55} 
F=\hat{\eta}(\bl) \, \T, \qquad F^\dagger =\TT \, \hat{\eta}^{-1}(\bl),
\eeq
so that
\beq\label{f56} F^\dagger F = T^\dagger T. \eeq
Hence, the inversion formula 
\begin{align}
	[F^\dagger F  \vf](\bx) = \phi(\bx)
\end{align}
again holds for $\phi(\bx) \in\S_{\bo, g}$. Note that since $| \hat{\eta}(\bl)| = 1$, estimates for the images $[F \phi](\bl)$ and $[\T\phi](\bl)$ coincide. Repeating arguments from the proof of Theorem~\ref{theoremf2}, from inversion formula we derive 
\begin{align}\label{f57} 
	\bigl\langle[F\vf_1](\bl),\,[F\vf_2](\bl) \bigr\rangle_{\hat{\Delta}} = \langle\vf_1(\bx), \, \vf_2(\bx)\rangle_\Delta, 
\end{align}
where $\phi_1, \phi_2 \in \SC_{\bo,g}$. 

Thus, in the same way, as before, the transform $F$ extends from the dense subset~$\mathcal{P}$~\eqref{Pset} to the isometry of $L_2$ spaces
\beq F \colon \; L_2^{\mathrm{sym}}(\mathbb{R}^n, {\Delta}) \; \to \; L_2^{\mathrm{sym}}(\mathbb{R}^n, \hat{\Delta}), \eeq
and therefore, as a consequence of bispectral duality~\eqref{bisp2}, it is unitary
\begin{align}
	F^\dagger F = \mathrm{Id}, \qquad F F^\dagger = \mathrm{Id}.
\end{align}

\subsection{Unitary isomorphisms}

We summarize the arguments above in the following theorem.

\begin{theorem}\label{theoremf3} In all four regimes I--IV the spectral transform $F$, given by scalar products~\rf{f48} and \rf{f54} with wave function $\Psi_{\bl}(\bx)$, admits the inversion formula
	\beq\label{f59} [F^\dagger F\vf](\bx)=\vf(\bx), \qquad \vf\in\S_{\bo, g}, \qquad \bx \in \mathbb{R}^n,\eeq
	and extends to the unitary isomorphism of Hilbert spaces
	\beqq L_2^{\mathrm{sym}}(\R^n , \mu)\qquad\text{and}  \qquad L_2^{\mathrm{sym}}(\R^n, \hat{\mu}),\eeqq
	or, correspondingly, 
	\beqq L_2^{\mathrm{sym}}(\R^n, \Delta)\qquad\text{and} \qquad L_2^{\mathrm{sym}}(\R^n, \hat{\Delta}).\eeqq
\end{theorem} 

Let us remark that in the particular case $g = \omega_2$ the Ruijsenaars system becomes ``free'': the Hamiltonians~\eqref{Hs} reduce to pure difference operators
\begin{align}
	H_s \bigr|_{g = \omega_2} =  (-1)^{s (n - s)} \, \sum_{ \substack{J \subset \{1, \dots, n\} \\[1pt] |J| = s}} \, \prod_{j \in J} e^{- \imath \omega_1 \partial_{x_j}}.
\end{align}
As shown by Halln\"as and Ruijsenaars~\cite[Theorem 3.1]{HR1}, the integral representing wave functions~\eqref{Psi} in this case can be calculated explicitly
\begin{align} \label{Psi-free}
	\Psi_{\bm{\lambda}}(\bm{x}) \bigr|_{g = \omega_2} = \prod_{1 \leq j < k \leq n} \frac{1}{4 \imath \sh \frac{\pi (x_j - x_k)}{\omega_1} \, \sh \pi \omega_1 (\lambda_j - \lambda_k) } \, \sum_{\sigma \in S_n} (-1)^{\sigma} \, e^{2\pi \imath \, \bm{\lambda} \cdot \sigma(\bm{x})}.
\end{align} 
Due to modular symmetry $\Psi_{\bm{\lambda}}(\bm{x}; g| \omega_1, \omega_2) = \Psi_{\bm{\lambda}}(\bm{x}; g| \omega_2, \omega_1)$, the same is true for $g = \omega_1$. 

The restriction $g = \omega_2$ is possible within regimes I and IV. Notice that in regime I the corresponding measures simplify (due to difference equation~\eqref{S2-diff})
\begin{align}
	&	\mu(\bm{x}) \bigr|_{g = \omega_2} = \prod_{\substack{j ,k = 1 \\ j \not=k}}^n \frac{S_2(\imath x_j - \imath x_k)}{S_2(\imath x_j - \imath x_k + \omega_2)} = \prod_{1 \leq j < k \leq n} 4 \sh^2 \frac{\pi (x_j - x_k)}{\omega_1}, \\[6pt]
	&	\hat{\mu}(\bm{\lambda})\bigr|_{g = \omega_2} = \prod_{1 \leq j < k \leq n} 4 \sh^2 \pi \omega_1 (\lambda_j - \lambda_k).
\end{align}
By absorbing half of these measures into the integral kernels of transforms $F$ and $F^\dagger$, given by~\eqref{Psi-free}, and the other half into functions, on which they act, one can show that $F$ and $F^\dagger$ are equivalent to multidimensional Fourier transform between spaces of square integrable \textit{antisymmetric} functions. The same is true regarding regime IV.

In the case of arbitrary coupling $g$, one can rescale by halfs of measure functions in a similar way, which leads to even more symmetry between spatial and spectral variables. Namely, consider any of the above unitarity regimes and denote the corresponding measure as $w(\bm{x})$. Let us introduce rescaled wave functions
\begin{align}
	\Phi_{\bm{\lambda}}(\bm{x}) = \sqrt{ \frac{w(\bm{x}) \, \hat{w}(\bm{\lambda}/\omega_1 \omega_2)}{(\omega_1 \omega_2)^n} } \, \Psi_{\bm{\lambda}/\omega_1 \omega_2}(\bm{x})
\end{align}
and the related spectral transforms
\begin{align} \label{Udef}
	\begin{aligned}
		& [U \phi](\bm{\lambda}) = \int_{\mathbb{R}^n} d\bm{x} \; \overline{\Phi_{\bm{\lambda}}(\bm{x}) } \, \phi(\bm{x}), \\[6pt]
		& [U^\dagger \chi](\bm{x}) = \int_{\mathbb{R}^n} d\bm{\lambda} \; \Phi_{\bm{\lambda}}(\bm{x}) \, \chi(\bm{\lambda}).
	\end{aligned}
\end{align}
Rewriting Theorem~\ref{theoremf3} in terms of rescaled functions we conclude that $U$ and $U^\dagger$ extend to unitary isomorphisms of the space $L^{\mathrm{sym}}_2(\mathbb{R}^n)$ (with trivial measure)
\begin{align} \label{UU}
	U^\dagger U = \mathrm{Id}, \qquad U U^\dagger = \mathrm{Id}.
\end{align}
Moreover, using bispectral duality~\eqref{bisp}, coupling reflection symmetry~\eqref{refl} and invariance of wave functions under rescaling of all parameters (see~\eqref{Psi} and~\eqref{S-hom})
\begin{align}
	\Psi_{\bm{\lambda}}(\bm{x}; g | \bm{\omega}) = \Psi_{\omega_1 \omega_2 \, \bm{\lambda}} \biggl( \frac{\bm{x}}{\omega_1 \omega_2}; \, \frac{g}{\omega_1 \omega_2} \, \bigg| \, \frac{\bm{\omega}}{\omega_1 \omega_2} \biggr),
\end{align}
one can check that
\begin{align}
	\Phi_{\bm{\lambda}}(\bm{x}) = \overline{ \Phi_{-\bm{x}}(\bm{\lambda}) }.
\end{align}
Then from definitions~\eqref{Udef} we have $U^\dagger = R U$, where $R$ is the reflection operator
\begin{align}
	[R \, \phi](\bm{x}) = \phi(-\bm{x}).
\end{align}
So, for rescaled transform the differences between functions of spatial and spectral variables disappear, and the formulas~\eqref{UU} are equivalent to the identity
\begin{align}
	U^2 = R.
\end{align}
This is analogous to well-known property of Fourier transform.

\subsection{Positive regularizations and Fatou's lemma} \label{sec:fatou}

In the previous work~\cite[Theorem~2]{BDKK3} we established unitarity of $F$ in regime I with the help of Fatou's lemma. The advantage of this approach is that one can work with compactly supported test functions $\phi(\bx)$ without requiring them to be analytic. In particular, by standard arguments~\cite[VIII.5]{KF} one can observe that in this case extension of the spectral transform to the corresponding $L_2$ space can be performed just using principal value integral.

On the other hand, Fatou's lemma can not be used outside of unitarity regimes, since it holds for nonnegative integrands only. Moreover, even in unitarity regimes one needs regularizing function~\eqref{R}
\begin{align}
	R_{\l, \epsilon}(\bm{\l})  = \exp \biggl( \pi  ( {g}^\ast - 2 \varepsilon ) \biggl[ n\l -\sum_{j = 1}^n \l_j \biggr] \biggr) \, \prod_{j = 1}^n \hat{K}(\l - \l_j),
\end{align}
which gives delta sequence~\eqref{f32}, to be nonnegative. Because this function is nonnegative in regime~II, the proof from~\cite{BDKK3} can be easily continued to this case as well. However, in regimes III, IV (complex coupling) the above function $R_{\l, \epsilon}(\bm{\l})$ is complex valued. 

A possible workaround is to choose different regularizing function, which will be nonnegative in regimes III, IV and such that the corresponding regularized scalar product can be again calculated explicitly. 

It turns out that we can simply take $| R_{\lambda, \epsilon}(\bl) |^2 \geq 0$. In regimes III, IV we have restriction $\bar{g} = g^*$, so that
\begin{align}
	| R_{\lambda, \epsilon}(\bl) |^2 = \exp \biggl( \pi  ( g + {g}^\ast - 4 \varepsilon ) \biggl[ n\l -\sum_{j = 1}^n \l_j \biggr] \biggr) \, \prod_{j = 1}^n \hat{K}(\l - \l_j) \, \hat{K}^*(\l - \l_j),
\end{align}
where two kernel functions from the right differ by reflection of the coupling $g \to g^*$
\begin{align}
	\hat{K}(\lambda) \equiv K(\lambda; \hat{g}^* | \hat{\bo} ), \qquad \hat{K}^*(\lambda) \equiv K(\lambda; \hat{g} | \hat{\bo} ).
\end{align}
Let us sketch the main steps of the unitarity proof with this regularization. The key idea is that the corresponding regularized scalar product
\begin{align}
	\bigl\langle \Psi_{\bm{\l}}(\by) ,\, \Psi_{\bm{\lambda}}(\bx) \bigr\rangle_{\hat{\Delta}}^{\l, \varepsilon}  = \int_{\mathbb{R}^n} d\bm{\l} \; \hat{\Delta}(\bm{\l}) \, \overline{\Psi_{\bm{\l}}(\bm{y})} \, \Psi_{\bm{\lambda}}(\bm{x})  \, | R_{\l, \epsilon}(\bm{\l}) |^2
\end{align}
can be reduced to the considered previously regularized pairing~\eqref{reg-pair}
\begin{align}
	\bigl( \Psi_{\bm{\l}}(\by) ,\, \Psi_{\bm{\lambda}}(\bx) \bigr)_{\hat{\mu}}^{\l, \varepsilon}  = \int_{\mathbb{R}^n} d\bm{\l} \; \hat{\mu}(\bm{\l}) \, \Psi_{-\bm{\l}}(\bm{y}) \, \Psi_{\bm{\lambda}}(\bm{x})  \, R_{\l, \epsilon}(\bm{\l})
\end{align}
with the help of Baxter $Q$-operator~\cite[Theorem~4]{BDKK4}
\begin{align}\label{Qdiag}
	Q^*(\lambda) \, \Psi_{\bl}(\bx) = \prod_{j = 1}^n \hat{K}^*(\lambda - \lambda_j) \, \Psi_{\bl}(\bx),
\end{align}
which is an integral operator acting on functions of $\bx$. Using Baxter operator~\eqref{Qdiag} and wave function properties after some calculations we obtain
\begin{align}
	\bigl\langle \Psi_{\bm{\l}}(\by) ,\, \Psi_{\bm{\lambda}}(\bx) \bigr\rangle_{\hat{\Delta}}^{\l, \varepsilon} = e^{\pi (g - 2\epsilon) n \lambda } \, \eta(\by) \, Q^*(\lambda) \, \bigl( \Psi_{\bm{\l}}(\tilde{\by}) ,\, \Psi_{\bm{\lambda}}(\bx) \bigr)_{\hat{\mu}}^{\l, \varepsilon} ,
\end{align}
where
\begin{align}
	\tilde{\bm{y}} = \bm{y} + \imath \biggl(\epsilon - \frac{g}{2} \biggr) \bm{e}, \qquad \bm{e} = (1, \dots, 1).
\end{align}
From that using the explicit formula for the regularized pairing~\eqref{f31} and Baxter operator~\cite[(1.29)]{BDKK4} we arrive at the expression for the scalar product
\begin{align}
	\begin{aligned}
		& \bigl\langle \Psi_{\bm{\l}}(\by) ,\, \Psi_{\bm{\lambda}}(\bx) \bigr\rangle_{\hat{\Delta}}^{\l, \varepsilon}  = (\omega_1 \omega_2)^{-n} \; \frac{\eta(\by)}{\eta(\bx)} \, \exp \biggl(2\pi \imath \, \lambda \sum_{j = 1}^n(x_j - y_j)\biggr) \\[6pt]
		& \qquad \times \int_{\mathbb{R}^n} d\bm{z} \; \Delta(\bm{z}) \, \prod_{j, k =1}^n K^*\biggl( x_j - z_k + \frac{\imath g}{2} - \imath \epsilon \biggr) \, K\biggl( y_j - z_k - \frac{\imath g^*}{2} + \imath \epsilon \biggr),
	\end{aligned}
\end{align}
where again kernel functions in the integrand differ by reflection $g \to g^*$
\begin{align}
	K(x) \equiv K(x; g | \bo), \qquad K^*(x) \equiv K(x; g^* | \bo).
\end{align}
In the limit $\epsilon \to 0^+$ the poles of the integrand pinch integration contours as long as $x_i = y_{j}$, so by careful analysis of residues and with the help of~\cite[Lemma~1]{BDKK3} one can prove that regularized scalar product is a delta sequence
\begin{align}\label{delta2}
	\lim_{\lambda \to \infty} \, \lim_{\epsilon \to 0^+} \bigl\langle \Psi_{\bm{\l}}(\by) ,\, \Psi_{\bm{\lambda}}(\bx) \bigr\rangle_{\hat{\Delta}}^{\l, \varepsilon} = \Delta^{-1}(\bx) \, \delta(\bx, \by).
\end{align} 
From this moment the proof of unitarity continues in the same fashion, as in~\cite{BDKK3}.

\setcounter{equation}{0}

\section*{Acknowledgments}
We are immensely grateful to S. Derkachov for collaboration and numerous helpful conversations which lead to the present work. We also appreciate G. Sarkissian, V. Spiridonov and Y. Lyubarskii for interesting discussions. S. Khoroshkin thanks Beijing Institute of Mathematical Sciences and Applications and the Technion for the kind hospitality. Sections 3 and 4 of this work were done within a research project implemented as part of the Basic Research Program at the HSE University.

\section*{Appendix}
\appendix

	\section{Double sine function}\label{AppendixA} 
	The double sine function $S_2(z) \equiv S_2(z|\omega_1, \omega_2)$, see \cite{Ku} and references therein, is a meromorphic function that satisfies two difference equations
	\beq\label{S2-diff}  \frac{S_2(z)}{S_2(z+\o_1)}=2\sin \frac{\pi z}{\o_2},\qquad \frac{S_2(z)}{S_2(z+\o_2)}=2\sin \frac{\pi z}{\o_1}
	\eeq
	and inversion relation
	\beq \label{S2-sin} S_2(z)S_2(-z)=-4\sin\frac{\pi z}{\o_1}\sin\frac{\pi z}{\o_2},\eeq
	or equivalently
	\beq\label{S2-refl} S_2(z)S_2(\o_1+\o_2-z)=1. \eeq
	This function can be expressed through the Barnes double gamma function $\Gamma_2(z|\bo)$ \cite{B}
	\beq
	S_2(z|\bo)=\Gamma_2(\o_1+\o_2-z|\bo) \, \Gamma_2^{-1}(z|\bo),
	\eeq
	and its properties follow from the corresponding properties of the double gamma function.
	It is also connected to the Ruijsenaars hyperbolic gamma function $G(z|\bo)$ \cite{R2}
	\beq \label{S-G}
	G(z|\bo) = S_2\Bigl(\imath z + \frac{\o_1 + \o_2}{2} \,\Big|\, \bo \Bigr)
	\eeq
	and to the Faddeev quantum dilogarithm $\gamma(z|\bo)$ \cite{F0}
	\beq 
	\gamma(z|\bo) = S_2\Bigl(-\imath z + \frac{\o_1+\o_2}{2}\, \Big|\, \bo\Bigr) \exp \biggl( \frac{\imath \pi}{2\o_1 \o_2} \biggl[z^2 + \frac{\o_1^2+\o_2^2}{12} \biggr]\biggr),
	\eeq
	which were investigated independently.

	The function $S_2(z)$ has poles and zeroes at the points
	\beq\label{S2-zeroes} z_{\mathrm{poles}} = \o_1 + \o_2 + m_1 \o_1 + m_2 \o_2, \qquad z_{\mathrm{zeroes}}=-m_1\o_1-m_2\o_2, \qquad m_i  \in \mathbb{N}_0.\eeq
	In the analytic region $ \Re z \in ( 0, \Re(\omega_1 + \omega_2) )$ we have the following integral representation for the logarithm of $S_2(z)$
	\begin{equation}\label{S2-int}
		\ln S_2 (z) = \int_0^\infty \frac{dt}{2t} \, \biggl( \frac{\sh \left[ (2z - \omega_1 - \omega_2)t \right]}{ \sh (\omega_1 t) \sh (\omega_2 t) } - \frac{ 2z - \omega_1 - \omega_2 }{ \omega_1 \omega_ 2 t } \biggr).
	\end{equation}
	It is clear from this representation that the double sine function is homogeneous
	\beq\label{S-hom}
	S_2( a z | a\o_1, a \o_2 ) = S_2(z|\o_1, \o_2), \qquad a \in (0, \infty)
	\eeq
	and invariant under permutation of periods
	\beq\label{A6}
	S_2(z| \o_1, \o_2) = S_2(z | \o_2, \o_1).
	\eeq
	Recall the multiple Bernoulli polynomial of the second order
	\begin{align}
		B_{2,2}(z | \bo) = \frac{1}{\omega_1 \omega_2} \biggl( \biggl[z - \frac{\omega_1 + \omega_2}{2} \biggr]^2 - \frac{\omega_1^2 + \omega_2^2}{12} \biggr).
	\end{align}
	Outside of cones with zeroes and poles~\eqref{S2-zeroes} double sine function has the following asymptotics
	\begin{align}\label{S2-asymp}
		S_2(z | \bo) \sim e^{\pm \frac{\pi \imath}{2} B_{2,2}(z | \bo)}, \qquad z \to \infty,
	\end{align}
	where the sign $+$ is taken for $z$ from the upper half plane and the sign $-$ for $z$ from the lower half plane.
	
	\section{Pairings} \label{AppendixB}
	
	The generating function of the Hamiltonians~\eqref{f1} can be written as the sum
	\begin{align}\label{H-HJ}
		H(\lambda) = e^{- n \pi \omega_1 \lambda} \sum_{J \subset \{1, \dots, n\}} \bigl(e^{2\pi \omega_1 \lambda} \bigr)^{n - |J|} \, H_J
	\end{align}
	of the difference operators
	\begin{align}
		H_J = \prod_{\substack{j \in J \\ k\not\in J}} \frac{\sh \frac{\pi}{\omega_2} (x_j - x_k - \imath g)}{ \sh \frac{\pi}{\omega_2} (x_j - x_k) } \, \prod_{j \in J} e^{ - \frac{\imath \omega_1}{2} \partial_{x_j} } \,  \prod_{k \not\in J} e^{ \frac{\imath \omega_1}{2} \partial_{x_k} }.
	\end{align}
	These operators are well defined on the functions $\phi(\bx)$ analytic in the strips
	\begin{align}\label{x-str3}
		| \Im x_j | \leq \frac{\Re \omega_1}{2}, \qquad j = 1, \dots, n.
	\end{align}
	Below we show that $H(\l)$ is symmetric with respect to the bilinear pairing
	\begin{align}\label{b-f}
		( \phi_1, \phi_2 )_\mu = \int_{\R^n} d\bx \; \mu(\bx) \, \phi_1(-\bx) \, \phi_2(\bx)
	\end{align}
	assuming analyticity and sufficient decay of the functions $\phi_i(\bx)$. Here the measure
	\begin{align}\label{mu2}
		\mu(\bx) = \frac{1}{n!} \prod_{\substack{j ,k = 1 \\ j \not=k}} \mu(x_j - x_k), \qquad \mu(x) = \frac{S_2(\imath x | \bo)}{S_2(\imath x + g | \bo)}.
	\end{align}
	The key ingredient of the proof is the difference equation
	\begin{align} \label{mu-eq}
		\frac{\mu(x - \imath \omega_1)}{\mu(x)} = \frac{\sh \frac{\pi (x - \imath g)}{\omega_2}}{\sh \frac{\pi x}{\omega_2}},
	\end{align}
	which in turn easily follows from the equation for the double sine function~\eqref{S2-diff}.
	
	\begin{lemma}\label{lemma:Hsym}
		Suppose $\phi_1(\bx), \phi_2(\bx)$ are analytic and admit the bound 
		\begin{align} \label{phi-b}
			| \phi_i(\bx) | \leq C \exp \biggl( - \pi \Re \hat{g} \sum_{1 \leq j < k \leq n} | \Re(x_j - x_k)| - \ve_i \sum_{j = 1}^n | \Re x_j | \biggr)
		\end{align}
		in the strips~\eqref{x-str3}, where $\ve_i \in \R$ are such that
		\begin{align}
			\ve_1 + \ve_2 > 0.
		\end{align}
		Then the operator $H(\l)$ is symmetric with respect to the bilinear pairing~\eqref{b-f} with these functions 
		\begin{equation}
			( \phi_1, \, H(\l) \, \phi_2)_\mu = ( H(\l) \, \phi_1, \, \phi_2)_\mu.
		\end{equation}
	\end{lemma}
	
	\begin{proof}
		Since $H(\l)$ splits into the sum of operators $H_J$~\eqref{H-HJ}, it is enough to prove the identity
		\begin{equation}\label{HJ-sym}
			( \phi_1, \, H_J \, \phi_2)_\mu = ( H_J \, \phi_1, \, \phi_2)_\mu
		\end{equation}
		for any $J \subset \{1, \dots, n\}$. Denote by $\bm{\xi}^J$ the vector with components
		\begin{align}\label{xi}
			j \in J \colon \quad \xi^J_j = \frac{\imath \omega_1}{2}, \qquad\quad k \not\in J \colon \quad  \xi^J_k = - \frac{\imath \omega_1}{2}.
		\end{align}
		Then the integrals from the identity~\eqref{HJ-sym} can be written in the following way
		\begin{align}\label{HJ-int}
			& ( \phi_1, \, H_J \, \phi_2)_\mu = \int_{\mathbb{R}^n} d\bm{x} \; \mu(\bm{x}) \, \phi_1(-\bm{x}) \,  \prod_{\substack{j \in J \\ k\not\in J}} \frac{\sh \frac{\pi}{\omega_2} (x_j - x_k - \imath g)}{ \sh \frac{\pi}{\omega_2}(x_j - x_k) } \, \phi_2(\bm{x} - \bm{\xi}^J), \\[6pt] \label{HJ-int2}
			& ( H_J \, \phi_1, \, \phi_2)_\mu = \int_{\mathbb{R}^n} d\bm{y} \; \mu(\bm{y}) \,  \prod_{\substack{j \in J \\ k\not\in J}} \frac{\sh \frac{\pi}{\omega_2} (y_j - y_k + \imath g)}{ \sh \frac{\pi}{\omega_2}(y_j - y_k) } \,  \phi_1(-\bm{y} - \bm{\xi}^J) \, \phi_2(\bm{y}).
		\end{align}
		To pass from the first integral to the second one we shift contours
		\begin{align}\label{x-shift}
			x_i \in \mathbb{R} \quad \to \quad x_i \in \mathbb{R} + \xi^J_i, \qquad i = 1, \dots, n,
		\end{align}
		change the integration variable
		\begin{align}
			\bm{x} = \bm{y} + \bm{\xi}^J 
		\end{align}
		and use the relation for the measure function
		\begin{align}\label{mu-eq2}
			\mu(\bm{y} + \bm{\xi}^J) \,  \prod_{\substack{j \in J \\ k\not\in J}} \frac{\sh \frac{\pi}{\omega_2} (y_j - y_k + \imath \omega_1  - \imath g)}{ \sh \frac{\pi}{\omega_2}(y_j - y_k + \imath \omega_1) } = \mu(\bm{y}) \,  \prod_{\substack{j \in J \\ k\not\in J}} \frac{\sh \frac{\pi}{\omega_2} (y_j - y_k + \imath g)}{ \sh \frac{\pi}{\omega_2}(y_j - y_k) } .
		\end{align}
		The last formula follows from the difference equation for the measure function of one variable~\eqref{mu-eq}.
		
		It is left to justify that all appearing integrals converge and we can shift contours in the prescribed way~\eqref{x-shift}. For this it is sufficient to show that the integrand in~\eqref{HJ-int}
		\begin{align} \label{HJ-intd}
			\mu(\bm{x}) \, \phi_1(-\bm{x}) \,  \prod_{\substack{j \in J \\ k\not\in J}} \frac{\sh \frac{\pi}{\omega_2} (x_j - x_k - \imath g)}{ \sh \frac{\pi}{\omega_2}(x_j - x_k) } \, \phi_2(\bm{x} - \bm{\xi}^J)
		\end{align}
		is analytic and absolutely integrable function in the strips
		\begin{align}\label{jk-strip}
			\begin{aligned}
				& j, j' \in J \colon  && \qquad 0 \leq \Im x_j  \leq \frac{\Re \omega_1}{2}, &&\qquad  | \Im (x_j - x_{j'}) | < \delta, \\[6pt]
				& k, k' \not\in J \colon &&\qquad  -\frac{\Re \omega_1}{2} \leq \Im x_k  \leq 0, &&\qquad  | \Im(x_k - x_{k'}) | < \delta
			\end{aligned}
		\end{align}
		for some small $\delta > 0$. First, notice that the functions $\phi_1(-\bx)$ and $\phi_2(\bx - \bm{\xi}^J)$ are analytic in the domain \eqref{jk-strip} by assumptions of this lemma. The measure function 
		\begin{align}
			\mu(\bm{x}) = \frac{1}{n!} \prod_{1 \leq j < k \leq n} 4 \sh \frac{\pi (x_j - x_k)}{\omega_1} \sh \frac{\pi (x_j - x_k)}{\omega_2} \, \prod_{\substack{j ,k = 1 \\ j \not=k}}^n S_2^{-1}(\imath x_j - \imath x_k + g)
		\end{align}
		has poles at the points
		\begin{align}
			x_j - x_k = \imath (g + m_1 \omega_1 + m_2 \omega_2), \qquad m_1, m_2 \in \mathbb{N}_0, \qquad j, k =1, \dots, n,
		\end{align}
		see \eqref{S2-zeroes}. For $j \in J$ and $k\not\in J$ all poles with $m_1 = 0$ cancel with the zeroes of the hyperbolic sines in~\eqref{HJ-intd}. And vice versa, the poles from the hyperbolic sines in denominators~\eqref{HJ-intd} cancel with the zeroes of the measure function. Altogether the function
		\begin{align}
			\mu(\bm{x}) \, \prod_{\substack{j \in J \\ k\not\in J}} \frac{\sh \frac{\pi}{\omega_2} (x_j - x_k - \imath g)}{ \sh \frac{\pi}{\omega_2}(x_j - x_k) } 
		\end{align}
		is analytic in the needed domain \eqref{jk-strip}. 
		
		Now let us show that the function~\eqref{HJ-intd} is also absolutely integrable for any fixed $\Im \bm{x}$ in the strips~\eqref{jk-strip}. In this domain the measure function together with the hyperbolic sines admits the bound
		\begin{align} \label{mu-sin-b}
			\Biggl|  \mu(\bm{x})  \prod_{\substack{j \in J \\ k\not\in J}} \frac{\sh \frac{\pi}{\omega_2} (x_j - x_k - \imath g)}{ \sh \frac{\pi}{\omega_2}(x_j - x_k) }  \Biggr| \leq C  \exp \biggl( 2 \pi \Re \hat{g} \sum_{1 \leq j < k \leq n} | \Re(x_j - x_k) |\biggr).
		\end{align} 
		This follows from the inequality \eqref{eta-bound} and simple estimate for the hyperbolic sines
		\begin{align}
			\Bigl| \sh \frac{\pi}{\omega_2} (x_j - x_k - \imath g) \Bigr| \leq C \exp \Bigl( \pi \Re \hat{\omega}_1 \, | \Re(x_j - x_k)|
			\Bigr).
		\end{align}
		By lemma assumptions~\eqref{phi-b} the rest of the integrand is also exponentially bounded
		\begin{multline} \label{phi2-b}
			| \phi_1(-\bx) \, \phi_2(\bx - \bm{\xi}^J) | \\[6pt]
			\leq C \exp \biggl( - 2 \pi \Re \hat{g} \sum_{1 \leq j < k \leq n } | \Re(x_j - x_k) | - (\ve_1 + \ve_2) \sum_{j = 1}^n | \Re x_j |\biggr),
		\end{multline}
		where $\ve_1 + \ve_2 > 0$. Combining~\eqref{mu-sin-b} and~\eqref{phi2-b} we obtain absolutely integrable bound for the integrand~\eqref{HJ-intd}.
	\end{proof}
	
	In Section~\ref{sec:un} we introduce four unitarity regimes for the parameters
	\begin{align*} 
		&\mathrm{I} \colon && \hspace{-2cm} \o_1,\o_2 > 0, && \hspace{-1cm} 0 < g < \omega_1 + \omega_2,\\[10pt]
		&\mathrm{II} \colon && \hspace{-2cm} \! \Re \omega_1 > 0, \; \omega_2 = \bar{\omega}_1, &&  \hspace{-1cm} 0 < g < \omega_1 + \omega_2,\\[8pt]
		&\mathrm{III} \colon && \hspace{-2cm} \o_1,\o_2 > 0,&& \hspace{-1cm} \Re g = \frac{\omega_1 + \omega_2}{2},\\[6pt]
		&\mathrm{IV} \colon && \hspace{-2cm} \! \Re \omega_1 > 0, \; \omega_2 = \bar{\omega}_1, && \hspace{-1cm} \Re g = \frac{\omega_1 + \omega_2}{2}
	\end{align*}
	and consider two sesquilinear pairings 
	\begin{align}\label{sp}
		\begin{aligned}
			& \langle \phi_1, \phi_2 \rangle_\mu = \int_{\mathbb{R}^n} d\bx \; \mu(\bx) \, \overline{\phi_1(\bx)} \, \phi_2(\bx),	\\[6pt]
			& \langle \phi_1, \phi_2 \rangle_\Delta = \int_{\mathbb{R}^n} d\bx \; \Delta(\bx) \, \overline{\phi_1(\bx)} \, \phi_2(\bx),
		\end{aligned}
	\end{align}
	where the first measure is the same, as in bilinear pairing above~\eqref{mu2}, while the second one is given by
	\begin{align}
 		\Delta(\bm{x}) = \frac{1}{n!} \prod_{1 \leq j < k \leq n} 4  \sh \frac{ \pi (x_j - x_k)}{{\omega}_1}\sh \frac{ \pi(x_j - x_k)}{{\omega}_2}.
	\end{align}
	Notice that the measures are symmetric in $\omega_1, \omega_2$, whereas the generating function of Hamiltonians $H(\lambda) \equiv H(\lambda| \omega_1, \omega_2)$~\eqref{H-HJ} is not. 
	
	\begin{lemma}\label{lemma:Hsym2}
		Assume the same conditions on $\phi_1, \phi_2$, as in Lemma~\ref{lemma:Hsym}. Then in the unitarity regimes I--IV the operator $H(\lambda | \omega_1, \omega_2)$ has the following symmetry properties with respect to the sesquilinear pairings~\eqref{sp}
		\begin{align}
			&\mathrm{I} \colon && \hspace{-2cm} \bigl\langle \phi_1, \, H(\lambda| \omega_1 ,\omega_2) \, \phi_2 \bigr\rangle_\mu =  \bigl\langle H(\lambda| \omega_1 ,\omega_2) \, \phi_1, \, \phi_2 \bigr\rangle_\mu, \\[6pt]
			&\mathrm{II} \colon && \hspace{-2cm} \bigl\langle \phi_1, \, H(\lambda| \omega_1 ,\omega_2) \, \phi_2 \bigr\rangle_\mu =  \bigl\langle H(\lambda| \omega_2 ,\omega_1) \, \phi_1, \, \phi_2 \bigr\rangle_\mu, \\[6pt]
			&\mathrm{III} \colon && \hspace{-2cm} \bigl\langle \phi_1, \, H(\lambda| \omega_1 ,\omega_2) \, \phi_2 \bigr\rangle_\Delta =  \bigl\langle H(\lambda| \omega_1 ,\omega_2) \, \phi_1, \, \phi_2 \bigr\rangle_\Delta, \\[6pt]
			&\mathrm{IV} \colon && \hspace{-2cm} \bigl\langle \phi_1, \, H(\lambda| \omega_1 ,\omega_2) \, \phi_2 \bigr\rangle_\Delta =  \bigl\langle H(\lambda| \omega_2 ,\omega_1) \, \phi_1, \, \phi_2 \bigr\rangle_\Delta.
		\end{align}
	\end{lemma}
	
	\begin{proof}
		The arguments are the same, as in the proof of Lemma~\ref{lemma:Hsym}. For the regimes I, II one uses difference equation on the measure $\mu(\by)$~\eqref{mu-eq2}, while in regimes III, IV there is analogous equation on the measure $\Delta(\by)$
		\begin{align}
			\Delta(\bm{y} + \bm{\xi}^J) \,  \prod_{\substack{j \in J \\ k\not\in J}} \frac{\sh \frac{\pi}{\omega_2} (y_j - y_k + \imath \omega_1  - \imath g)}{ \sh \frac{\pi}{\omega_2}(y_j - y_k + \imath \omega_1) } = \Delta(\bm{y}) \,  \prod_{\substack{j \in J \\ k\not\in J}} \frac{\sh \frac{\pi}{\omega_2} (y_j - y_k + \imath g^*)}{ \sh \frac{\pi}{\omega_2}(y_j - y_k) } ,
		\end{align}
		where vector $\bm{\xi}^J$ is defined in~\eqref{xi}. 
	\end{proof}


\begin{thebibliography}{99}
		
		\bibitem[AAR]{AAR} G. E. Andrews, R. Askey, R. Roy, \textit{Special Functions}, Cambridge University Press (1999).
		
		\bibitem[AOS]{AOS} H. Awata, S. Odake, J. Shiraishi, \textit{Integral representations of the Macdonald symmetric polynomials}, \href{https://doi.org/10.1007/BF02100101}{Communications in Mathematical Physics} \textbf{179} (1996) 647--666, \href{https://doi.org/10.48550/arXiv.q-alg/9506006}{\tt arXiv:q-alg/9506006}.
		
		\bibitem[B]{B} E. W. Barnes, \textit{The theory of the double gamma function}, \href{https://www.jstor.org/stable/90809}{Philosophical Transactions of the Royal Society of London}. Series A, Containing Papers of a Mathematical or Physical Character \textbf{196} (1901) 265--387.
		
		\bibitem[BDKK1]{BDKK1} N. Belousov, S. Derkachov, S. Kharchev, S. Khoroshkin, \textit{Baxter operators in Ruijsenaars hyperbolic system I: commutativity of $Q$-operators}, \href{https://doi.org/10.1007/s00023-023-01364-4}{Annales Henri Poincaré} \textbf{25} (2024) 3207--3258, \href{https://doi.org/10.48550/arXiv.2303.06383}{\tt arXiv:2303.06383 [math-ph]}.
		
		\bibitem[BDKK2]{BDKK2} N. Belousov, S. Derkachov, S. Kharchev, S. Khoroshkin, \textit{Baxter operators in Ruijsenaars hyperbolic system II: bispectral wave functions}, \href{https://doi.org/10.1007/s00023-023-01385-z}{Annales Henri Poincaré} \textbf{25} (2024) 3259--3296, \href{https://doi.org/10.48550/arXiv.2303.06382}{\tt arXiv:2303.06382 [math-ph]}.
		
		\bibitem[BDKK3]{BDKK3} N. Belousov, S. Derkachov, S. Kharchev, S. Khoroshkin, \textit{Baxter operators in Ruijsenaars hyperbolic System III: orthogonality and completeness of wave functions}, \href{https://doi.org/10.1007/s00023-023-01406-x}{Annales Henri Poincaré} \textbf{25} (2024) 3297--3332, \href{https://doi.org/10.48550/arXiv.2307.16817}{\tt arXiv:2307.16817 [math-ph]}.
		
		\bibitem[BDKK4]{BDKK4} N. Belousov, S. Derkachov, S. Kharchev, S. Khoroshkin, \textit{Baxter operators in Ruijsenaars hyperbolic system IV: coupling constant reflection symmetry}, \href{https://doi.org/10.1007/s00220-024-04952-5}{Communications in Mathematical Physics} \textbf{405} (2024), \href{https://doi.org/10.48550/arXiv.2308.07619}{\tt arXiv:2308.07619 [math-ph]}.
		
		\bibitem[BKK]{BKK} M. Bullimore, H. C. Kim, P. Koroteev, \textit{Defects and quantum Seiberg-Witten geometry}, \href{https://doi.org/10.1007/JHEP05(2015)095}{Journal of High Energy Physics} \textbf{2015} (2015) 95, \href{https://doi.org/10.48550/arXiv.1412.6081}{\tt arXiv:1412.6081 [hep-th]}.
		
		\bibitem[BT]{BT} A. G. Bytsko, J. Teschner, \textit{R-operator, co-product and Haar-measure for the modular double of $U_q(\mathfrak{sl}(2,\mathbb{R}))$}, \href{https://doi.org/10.1007/s00220-003-0894-5}{Communications in mathematical physics} \textbf{240} (2003) 171--196, \href{https://doi.org/10.48550/arXiv.math/0208191}{\tt arXiv:math/0208191 [math.QA]}.
				
		\bibitem[DKKSS]{DKKSS} P. Di Francesco, R. Kedem, S. Khoroshkin, G. Schrader, A. Shapiro, \textit{Ruijsenaars wavefunctions as modular group matrix coefficients}, \href{https://doi.org/10.48550/arXiv.2402.14214}{\tt arXiv:2402.14214 [math-ph]} (2024).
		
		\bibitem[DKM1]{DKM1} S. E. Derkachov, K. K. Kozlowski, A. N. Manashov, On the separation of variables for the modular XXZ magnet and the lattice Sinh-Gordon models, \href{https://doi.org/10.1007/s00023-019-00806-2}{Annales Henri Poincaré} \textbf{20} (2019) 2623--2670, \href{https://doi.org/10.48550/arXiv.1806.04487}{\tt arXiv:1806.04487 [math-ph]}.
		
		\bibitem[DKM2]{DKM} S. E. Derkachov, K. K. Kozlowski, A. N. Manashov, \textit{Completeness of SoV Representation for $SL(2,\mathbb{R})$ Spin Chains}, \href{https://doi.org/10.3842/SIGMA.2021.063}{SIGMA} \textbf{17} (2021) 063, \href{https://doi.org/10.48550/arXiv.2102.13570}{\tt arXiv:2102.13570 [math-ph]}.
		
		\bibitem[F1]{F0} L. D. Faddeev, \textit{Discrete Heisenberg-Weyl Group and modular group}, \href{https://doi.org/10.1007/BF01872779}{Letters in Mathematical Physics} \textbf{34} (1995) 249--254, \href{https://doi.org/10.48550/arXiv.hep-th/9504111}{\tt arXiv:hep-th/9504111}.
		
		\bibitem[F2]{F1} L. D. Faddeev, \textit{Modular double of quantum group}, Mathematical Physics Studies \textbf{21} (2000), \href{https://doi.org/10.48550/arXiv.math/9912078}{\tt arXiv:math/9912078 [math.QA]}.
		
		\bibitem[F3]{F2} L. D. Faddeev, \textit{Discrete series of representations for the modular double of $U_q (\mathrm{sl} (2, \mathbb{R}))$}, \href{https://doi.org/10.48550/arXiv.0712.2747}{\tt arXiv:0712.2747 [math.QA]} (2007).
		
		\bibitem[HR1]{HR1} M. Halln\"as, S. Ruijsenaars, \textit{Joint eigenfunctions for the relativistic Calogero--Moser Hamiltonians of hyperbolic type. I. First steps}, \href{https://doi.org/10.1093/imrn/rnt076}{International Mathematics Research Notices} \textbf{2014}:16 (2014) 4400--4456, \href{https://doi.org/10.48550/arXiv.1206.3787}{\tt arXiv:1206.3787 [nlin.SI]}.
		
		\bibitem[HR2]{HR2} M. Halln\"as, S. Ruijsenaars, \textit{Joint eigenfunctions for the relativistic Calogero--Moser Hamiltonians of hyperbolic type. II. The two- and three-variable cases}, \href{https://doi.org/10.1093/imrn/rnx020}{International Mathematics Research Notices} \textbf{2018}:14 (2018) 4404--4449, \href{https://doi.org/10.48550/arXiv.1607.06672}{\tt arXiv:1607.06672 [math-ph]}.
		
		\bibitem[HR3]{HR3} M. Halln\"as, S. Ruijsenaars, \textit{Joint eigenfunctions for the relativistic Calogero--Moser Hamiltonians of hyperbolic type. III. Factorized asymptotics}, \href{https://doi.org/10.1093/imrn/rnaa193}{International Mathematics Research Notices} \textbf{2021}:6 (2021) 4679--4708, \href{https://doi.org/10.48550/arXiv.1905.12918}{\tt arXiv:1905.12918 [math-ph]}.
		
		\bibitem[KF]{KF} A. N. Kolmogorov, S. V. Fomin, \textit{Elements of the theory of functions and functional analysis} (Russian), Forth edition, Moscow, Nauka (1976).
		
		\bibitem[KK]{Ku} N. Kurokawa, S-Y. Koyama, \textit{Multiple sine functions}, \href{https://doi.org/10.1515/form.2003.042}{Forum Mathematicum} \textbf{15} (2003) 839--876.		
		
		\bibitem[KLS]{KLS} S. Kharchev, D. Lebedev, M. Semenov-Tian-Shansky, \textit{Unitary representations of $U_q(\mathfrak{sl}(2, \mathbb{R}))$, the modular double and the multiparticle $q$-deformed Toda chain}, \href{https://doi.org/10.1007/s002200100592}{Communications in mathematical physics} \textbf{225} (2002) 573--609, \href{https://doi.org/10.48550/arXiv.hep-th/0102180}{\tt arXiv:hep-th/0102180}.
		
		\bibitem[M]{M} I. Macdonald, \textit{Symmetric functions and Hall Polynomials}, Second edition, Oxford, Oxford University Press (1995).
		
		\bibitem[O]{O} E. M. Opdam, \textit{Harmonic analysis for certain representations of graded Hecke algebras}, \href{https://doi.org/10.1007/BF02392487}{Acta Mathematica} \textbf{175} (1995) 75--121.
		
		\bibitem[R1]{R1} S. N. M. Ruijsenaars, \textit{Complete integrability of relativistic Calogero-Moser systems and elliptic function identities}, \href{https://doi.org/10.1007/BF01207363}{Communications in Mathematical Physics} \textbf{110} (1987) 191--213.
		
		\bibitem[R2]{R2} S. N. M. Ruijsenaars, \textit{First-order analytic difference equations and integrable quantum systems}, \href{https://doi.org/10.1063/1.531809}{Journal of Mathematical Physics} \textbf{38} (1997) 1069--1146.
		
		\bibitem[R3]{R3} S. N. M. Ruijsenaars, A relativistic conical function and its Whittaker limits, \href{https://doi.org/10.3842/SIGMA.2011.101}{SIGMA} \textbf{7} (2011) 101, \href{https://doi.org/10.48550/arXiv.1111.0115}{\tt arXiv:1111.0115 [math.CA]}.
					
		\bibitem[SS]{SS} G. A. Sarkissian, V. P. Spiridonov, \textit{Complex hypergeometric functions and integrable many-body problems}, \href{https://doi.org/10.1088/1751-8121/ac88a4}{Journal of Physics A: Mathematical and Theoretical}, \textbf{55} (2022), \href{https://doi.org/10.48550/arXiv.2105.15031}{\tt arXiv:2105.15031 [math-ph]}.
		
		\bibitem[SSh]{SchSh} G. Schrader, A. Shapiro, On $b$-Whittaker functions, \href{https://doi.org/10.48550/arXiv.1806.00747}{\tt arXiv:1806.00747 [math-ph]} (2018).
		
		\bibitem[T]{T} J. Teschner, \textit{From Liouville theory to the quantum geometry of Riemann surfaces}, in 14th International Congress on Mathematical Physics (2003), \href{https://doi.org/10.48550/arXiv.hep-th/0308031}{\tt arXiv:hep-th/0308031}.
		
		\bibitem[TV]{TV} J. Teschner, G. S. Vartanov, \textit{Supersymmetric gauge theories, quantization of $\mathcal{M}_{\mathrm {flat}} $, and conformal field theory}, \href{https://dx.doi.org/10.4310/ATMP.2015.v19.n1.a1}{Advances in Theoretical and Mathematical Physics} \textbf{19}:1 (2015) 1--135, \href{https://doi.org/10.48550/arXiv.1302.3778}{\tt arXiv:1302.3778 [hep-th]}.

	\end{thebibliography}
\end{document}